\newcommand{\full}{}

\def\year{2017}\relax
\documentclass[letterpaper]{article}
\usepackage{aaai17}
\usepackage{times}
\usepackage{helvet}
\usepackage{courier}
\usepackage{soul}
\usepackage[normalem]{ulem}
\usepackage{url}
\usepackage{array}
\usepackage{amssymb}
\usepackage{amsmath}
\usepackage{amsthm}
\usepackage{thm-restate}
\usepackage{enumitem}
\usepackage{algorithm}
\usepackage{graphicx}
\usepackage[scriptsize,tight]{subfigure}
\usepackage{algorithmicx,algpseudocode}
\usepackage{tikz}   % drawing package
\usetikzlibrary{calc}
\usepackage{epstopdf}   % Compile with pdflatex even if only .eps fig is given
\usepackage{ifthen}
\usepackage{sistyle}
\SIthousandsep{,}
\usepackage{float}
\graphicspath{{figure/}}

\DeclareMathOperator{\E}{\mathbb{E}}
\renewcommand{\Pr}{\mathbb{P}}
\DeclareMathOperator{\Var}{\mathrm{Var}}
\DeclareMathOperator{\Cor}{\mathrm{Cor}}

\DeclareMathOperator{\Cut}{\mathrm{Cut}}

\DeclareMathOperator{\diag}{\mathrm{diag}}
\DeclareMathOperator{\Vol}{\mathrm{Vol}}

\newtheorem*{proposition*}{Proposition}

\newtheorem{definition}{Definition}
\newcommand\numberthis{\addtocounter{equation}{1}\tag{\theequation}}

\newcommand{\e}{\mu}    % expected value
\renewcommand{\c}{\rho}         % correlation
\newcommand{\s}{\sigma}         % similarity
\newcommand{\model}{VIO}

\DeclareMathOperator*{\argmin}{arg\,min}

\ifdefined\full
\nocopyright
\addtocounter{page}{2}
\fi

\newcommand{\citet}[1]{\citeauthor{#1} \shortcite{#1}}

\begin{document}
    \ifdefined\full
    \title{Partitioned Sampling of Public Opinions \\
        Based on Their Social Dynamics\protect\thanks{\noindent This work was supported in part by the National Basic Research Program of China Grant 2011CBA00300, 2011CBA00301, 
            the National Natural Science Foundation of China Grant 61033001, 61361136003, 61433014.}}
    \else
    \title{Partitioned Sampling of Public Opinions \\
        Based on Their Social Dynamics\protect\thanks{\noindent This work was supported in part by the National Basic Research Program of China Grant 2011CBA00300, 2011CBA00301, 
            the National Natural Science Foundation of China Grant 61033001, 61361136003, 61433014.}}
    \fi
    \author{Weiran Huang\\
        IIIS, Tsinghua University\\
        Beijing, China\\
        \textsf{huang.inbox@outlook.com}\\
        \And Liang Li\\
        AI Department, Ant Financial Group\\
        Hangzhou, Zhejiang, China\\
        \textsf{liangli.ll@alipay.com}\\
        \And Wei Chen\\
        Microsoft Research\\
        Beijing, China\\
        \textsf{weic@microsoft.com}\\
    }
    
    \maketitle
    \begin{abstract}
        Public opinion polling is usually done by random sampling from the entire population, 
        treating individual opinions as independent.
        In the real world, individuals' opinions are often correlated, e.g., among friends in a social network.
        In this paper, we explore the idea of partitioned sampling, 
        which partitions individuals with high opinion similarities into groups
        and then samples every group separately to obtain an accurate estimate of the population opinion.
        We rigorously formulate the above idea as an optimization problem.
        We then show that the simple partitions which contain only one sample in each group are 
        always better, and 
        reduce finding the optimal simple partition to a well-studied Min-$r$-Partition problem.
        We adapt an approximation algorithm and a heuristic %algorithm 
        to solve the optimization problem.
        Moreover, to obtain opinion similarity efficiently,
        we adapt a well-known opinion evolution model to characterize social interactions, 
        and provide an exact computation of opinion
        similarities based on the model.
        We use both synthetic and real-world datasets to demonstrate that 
        the partitioned sampling method results in significant improvement in sampling quality
        and it is robust when some opinion similarities are inaccurate or even missing.
    \end{abstract}
    
    \section{Introduction}
    
    Public opinion is essential nowadays for governments, organizations and companies to make decisions on their policies, strategies, products, etc.
    The most common way to collect public opinions is polling, typically
    done by randomly sampling a large number of individuals from the entire population and then interviewing them by telephone.
    This naive method is unbiased, but conducting interviews is very costly.
    On the other hand, in recent years, more and more online social media data are available and have been used to predict public opinions on certain issues.
    Such predictions cost less human effort, but they are usually biased
    and may lead to incorrect decisions.
    Thus, keeping the estimation unbiased while saving the cost becomes an important task to pursue.
    
    In this paper, we utilize individuals' social interactions
    (potentially learned from social media data) 
    to improve the unbiased sampling method.
    Our motivation is from the fact that 
    people's opinions are often correlated, especially among friends in a social network, 
    due to their social interactions in terms of the homophily and influence effects \cite{Homophily,Duncan,Kleinberg}.
    Such correlations are partially known in the big data era.
    For example, many online social media and networking sites provide publicly available social interaction data and user's sentiment data, and companies also have large amounts of data about their customers' preferences and their social interactions. 
    Our idea is to partition individuals into different groups by utilizing the above prior knowledge,
    such that people within a group are likely to hold the same opinions.
    We can then sample very few people in each group and aggregate the sampling results together to achieve an accurate estimation.
    We call this the {\em partitioned sampling} method.
    
    We formulate the above idea as an optimization problem.
    In particular, we first 
    characterize individuals' opinions as random variables. 
    We then specify our objective as minimizing
    the expected sample variance of the estimate, and
    define the statistical measure of pairwise {\em opinion similarity} as the input.
    Our analysis later shows that this input is enough to fully determine the solution
    of the optimization problem, named the
    {\em Optimal Partitioned Sampling (OPS)} problem (Section \ref{section: formulating the problem}).
    
    We solve the OPS problem in two steps (Section \ref{section: solving the OPS problem}).
    First, we show that 
    the best partition is always a {\em simple partition}, 
    meaning that each group contains only one sample.
    Second, we use people's opinion similarities to construct a weighted graph and 
    reduce the OPS problem to the Min-$r$-Partition problem.
    We adapt a semi-definite programming algorithm and a heuristic algorithm
    to solve the optimization problem.
    We further show that partitioned sampling using any balanced simple partition where group sizes are the same
    always out-performs naive sampling method, and thus balanced simple partition is always safe to use even if we only have
    partial or inaccurate opinion similarity information.
    
    Next, we adapt existing opinion evolution models and propose the Voter model with Innate Opinions (VIO) based
    on social network interactions (Section \ref{section: opinion evolution model and opinion similarity}).
    We provide an exact computation of opinion similarities in the steady state of
    the model, which is novel in the study of such models.
    
    Finally, we conduct experiments on both synthetic and real-world datasets
    to demonstrate the effectiveness and robustness of our partitioned sampling method
    (Section \ref{section: experimental evaluation}).
    
    In summary, our contributions include: 
    (a) proposing the partitioned sampling method to
    improve sampling quality based on opinion similarities
    and formulating it as an optimization problem, 
    (b) precisely connecting the OPS problem to
    the Min-$r$-Partition problem and providing efficient algorithms 
    for the OPS problem,
    and 
    (c) adapting an opinion evolution model and providing
    an exact computation of opinion similarities based on the model.
    
    \ifdefined\full
    \else
    Due to space constraints, all proofs and further technical details are included in
    the full report \cite{HLC15}.
    \fi
    
    \textbf{Related Work.}
    There are many sampling methods in the literature.
    The most related method is stratified sampling
    \cite{bethel1986optimum,bethel1989sample,chromy1987design,cochran2007sampling,kozak2007modern,keskinturk2007genetic,ballin2013joint}.
    The entire population is first stratified into homogeneous atomic strata based on individuals' profiles (e.g., age, gender, etc.),
    and then they may be combined to a final stratification and 
    subsample size in each stratum is allocated to minimize sample variance.
    Conceptually, our partitioned sampling method is similar to stratified sampling, but there are some important differences.
    First, stratified sampling partitions individuals based on their profiles, which may not imply opinion similarity, while we partition individuals
    directly based on opinion similarity, and thus our method is more accurate and flexible.
    Second, the technical treatments are different. Stratified sampling treats individual opinions as fixed and unknown, 
    and requires the (estimated) mean and standard deviation of opinions in each stratum to bootstrap the
    stratified sampling, while we treat individual opinions as random variables, and use pairwise opinion 
    similarities for partitioned sampling.
    
    Among studies on social interaction based sampling, \citet{SocialSampling} utilize social network connections to facilitate sampling.
    However, their method is to ask the voter being sampled to return the estimate of her friends' opinions,
    which changes the polling practice. 
    In contrast, we still follow the standard polling practice and only use implicit knowledge on opinion similarities to improve sampling quality.
    \citet{DebiasingSocial} consider the task of estimating 
    people's average innate opinion by removing their social interactions,
    which is opposite to our task --- we want to utilize opinion interactions for more efficient sampling of final expressed opinions which are counted in opinion polls.
    Graph sampling methods \cite{gjoka2010walking,kurant2011walking} aim at achieving unbiased uniform sampling on
    large scale networks when the full network is not available, which is orthogonal to our partitioned sampling
    approach and could be potentially combined. 
    
    Various opinion evolution models have been proposed in the literature
    \cite{DiscreteOpinion,DebiasingSocial,OpinionMaximization,Li0WZ15}.
    Our \model{} model is adapted from the voter model \cite{VoterModel} and its extension with stubborn agents \cite{DiscreteOpinion}.
    
    Graph partitioning has been well studied, and numerous problem variants and algorithms exist.
    In this paper, we reduce the OPS problem to the Min-$r$-Partition problem, 
    which was first formulated by \citet{sahni1976p}.
    To the best of our knowledge, there is no approximation or heuristic algorithms for Min-$r$-Partition.
    Thus, we adapt a state-of-art approximation algorithm for the dual problem (Max-$r$-Cut) to solve the OPS problem \cite{maxkcut}.
    We also propose a greedy algorithm for large graphs, which takes the idea from a heuristic algorithm for Max-$r$-Cut \cite{zhu2013max}.
    
    \section{Formulating the OPS Problem}\label{section: formulating the problem}
    
    We consider a vertex set $V$ from a social network graph 
    containing $n$ vertices (or nodes) $v_1, v_2, \dots, v_n$.
    Each vertex represents a person in the social network, 
    and has a binary opinion on some topic of interest. 
    Our task is to estimate the average opinion of all individuals in the social network with sample size budget $r$. 
    Let $f:V\rightarrow\{0, 1\}$ denote the opinion function, 
    i.e., we wish to estimate the fraction $\bar f = \frac{1}{n} \sum_{i=1}^n f(v_i)$.
    The {\em naive sampling} method simply picks $r$ nodes uniformly at random with replacement 
    from $V$ to ask their opinions 
    and takes the average of sampled opinions as the estimate, as denoted below:
    $\hat f_{\it naive} (V,r) = \frac{1}{r} \sum_{i=1}^r f(x_i)$, where $x_i$ is the $i$-th sampled node.
    
    In this paper, we propose a general sampling framework called {\em partitioned sampling}.
    Formally, we first partition the whole vertex set into several disjoint subsets (called {\em groups}), 
    and then allocate subsample size of each group.
    We use $\mathcal{P}=\{(V_1,r_1),$ $(V_2,r_2), \dots, (V_K,r_K)\}$ to represent such a partition, 
    where $V_1$, $V_2$, \dots, $V_K$ are groups, 
    and $r_k$ is the subsample size of group $V_k$.
    Next, we do naive sampling inside each group $V_k$ with its subsample size $r_k$.
    Finally, we estimate the average opinion of the population
    by taking a weighted average of all subsampling results, 
    with weights proportional to group sizes:  
    $\hat f_{\it part}({\cal P}) =  \sum_{k=1}^K \frac{|V_k|}{|V|} \cdot \hat f_{\it naive} (V_k,r_k)$.
    Notice that naive sampling is a special case of partitioned sampling with $\mathcal{P}=\{(V,r)\}$.
    One can easily verify that partitioned sampling is unbiased 
    \ifdefined\full
    (see Appendix \ref{unbiased proof}).
    \else
    \cite{HLC15}.
    \fi
    
    Intuitively, the advantage of using partitioned sampling is that, 
    if we partition individuals such that people likely holding the same opinions are partitioned into the same group, 
    then we can sample very few people in each group to get an accurate estimate of the average opinion of the group, 
    and aggregate them to get a good estimate of population mean.
    To implement this idea, we assume that some prior knowledge about people's opinions and their
    similarities is available before sampling.
    Based on these knowledge, our goal is to find the best partition for partitioned sampling which achieves the best sampling quality.
    
    Our first research challenge is how to rigorously formulate the above intuition into an optimization problem.
    To meet this challenge, we need to answer (a) which objective function is the appropriate one for the optimization problem, 
    and (b) which representation of the prior knowledge about people's opinions and their similarities 
    can be used as the inputs to the optimization problem.
    
    We first address the objective function.
    When all individuals' opinions $f(v_1), f(v_2), \dots, f(v_n)$ are fixed (but unknown), 
    the effectiveness of an unbiased randomized sampling method is measured by the standard sample variance $\Var(\hat{f})$, 
    where $\hat{f}$ is the estimate.
    The smaller the sample variance, the better the sampling method.
    When the prior statistical knowledge about people's opinions is available, effectively we treat
    opinions $f(v_1), f(v_2), \dots, f(v_n)$ as random variables, and the prior knowledge is
    some statistics related to the joint distribution of these random variables.
    In this case, the best sampling method should minimize the {\em expected sample variance} $\E[\Var(\hat{f})]$, 
    where the expectation is taken over the randomness from the joint distribution of people's opinions.
    For clarity, we use $\E_M[\Var_S(\hat{f})]$ to represent $\E[\Var(\hat{f})]$, where
    subscript $M$ (standing for ``model'') represents the randomness from the
    joint distribution model of opinions, and subscript $S$ (standing for ``sampling'') represents sample randomness
    from the sampling method.
    \footnote{One may propose to use the total variance $\Var_{M,S}(\hat f)$ as the objective function. 
        \ifdefined\full
        In Appendix \ref{section:object discussion}, we show that they are equivalent for the optimization task.
        \else
        In the full report \cite{HLC15}, we show that they are equivalent for the optimization task.
        \fi}

    We now discuss the input to the optimization task. 
    The full joint distribution of $f(v_1), f(v_2), \dots, f(v_n)$ requires
    an exponential number of parameters and is infeasible as the input.
    Then notice that the objective function only involves first two moments, which suggests us to use the expectations and pairwise correlations of people's opinions 
    as the inputs.
    Indeed, we find that these knowledge is good enough to fully characterize the optimization problem. %we need to solve.
    However, we further discover that a weaker and more direct type of statistics would be 
    enough to enable the 
    optimization problem, which we formally define as pairwise opinion similarities:
    the {\em opinion similarity} $\sigma_{ij}$ for nodes $v_i$ and $v_j$ is defined as the probability that 
    $f(v_i)$ and $f(v_j)$ have the same values.
    
    With the objective function and inputs settled, we are now ready to define our optimization problem:
    
    \begin{definition}\label{def:OPS}
        (Optimal Partitioned Sampling) 
        Given a vertex set $V=\{v_1, v_2, \dots, v_n\}$, sample size budget $r<n$, and
        opinion similarity $\s_{ij}$ between every pair of nodes $v_i$ and $v_j$,
        the {\em Optimal Partitioned Sampling} (OPS) problem is to find the optimal partition $\mathcal{P}^*$ of $V$, 
        %with the corresponding sample size allocation, 
        such that the
        partitioned sampling method % as given in Method~\ref{al:partitioned}
        using $\mathcal{P}^*$ achieves the minimum 
        expected sample variance, i.e.,  $\mathcal{P}^* = \argmin_{\cal P} \E_M[\Var_S(\hat f_{\it part}({\cal P}))]$, where $\cal P$ takes among all partitions of
        $V$ with $r$ samples.
    \end{definition}
    
    We remark that the OPS problem requires all pairwise opinion similarities as inputs so as to make the problem well-defined.
    We will address the issue of handling missing or inaccurate opinion similarities 
    in Section \ref{section: discussion of similarity}, and show that partitioned sampling still has outstanding performance.
    
    \section{Solving the OPS Problem}\label{section: solving the OPS problem}
    
    There are two issues involved in the OPS problem: 
    one is how to partition the vertex set $V$ into $K$ groups; 
    the other is how to allocate the subsample size in each group.
    For simplifying the OPS problem, 
    we first consider a special kind of partitions that pick only one sample node in each group.
    
    \begin{definition}
        A {\em simple partition} is a partition in which the subsample size of each group is equal to one.
    \end{definition}
    
    Simple partitions are important not only for the simplicity but also for the superiority.
    We will later show in Theorem \ref{theorem: simple is better} that, for any non-simple partition $\cal P$,
    one can easily construct a simple partition based on $\cal P$ which is at least as good as $\cal P$.
    Thus, we focus on finding the optimal simple partition.% for the OPS problem.
    
    Our approach is constructing a weighted assistant graph $G_a$ whose vertex set is $V$,
    where the weight of edge $(v_i,v_j)$ is $w_{ij}=1-\s_{ij}$,
    and then connecting the OPS problem with a graph partitioning problem for the graph $G_a$.
    For a simple partition $\mathcal{P}=\{(V_1,1), (V_2,1), \dots, (V_r,1)\}$ of $V$, 
    we use $\Vol_{G_a}(V_k)$ to denote the volume of the group $V_k$ in the graph $G_a$, 
    defined as $\Vol_{G_a}(V_k)=\sum_{v_i,v_j\in V_k} w_{ij}$.
    We define a cost function $g(\mathcal{P})$ to be the sum of all groups' volumes in $G_a$, 
    namely, $g(\mathcal{P})=\sum_{k=1}^r \Vol_{G_a}(V_k)$.
    Our major technical contribution is to show that
    minimizing the expected sample variance of partitioned sampling using any simple partition $\mathcal{P}$ 
    is equivalent to 
    minimizing the cost function $g(\mathcal{P})$, 
    as summarized by the following theorem:
    
    \begin{restatable}{theorem}{objective}\label{theorem:objective function}
        Given a vertex set $V$ with pairwise opinion similarities $\{\s_{ij}\}$'s and sample size $r$, 
        for any simple partition ${\cal P}=\{(V_1,1), (V_2,1),$ $ \dots, $ $(V_r,1)\}$ of $V$, 
        $$\E_M[\Var_S(\hat f_{\it part}({\cal P}))] =g({\mathcal{P}})/2|V|^2.$$ 
        Thus, the optimal simple partition of $V$ minimizes the cost function $g({\mathcal{P}})$.
    \end{restatable}
    \begin{proof}[Proof (Sketch)]
        We use $x_k$ to denote the sample node selected in the $k$-th group $V_k$ of the simple partition $\mathcal{P}$. 
        The estimate of partitioned sampling with $\mathcal{P}$ can be written as
        $\hat f_{\it part}(\mathcal P) = \frac{1}{n}\sum_{k=1}^r n_k f(x_k)$,
        where $n=|V|$ and $n_k=|V_k|$.
        When $f$ is fixed, since $f(x_k)$'s are independent, then
        \begin{align*}
            &\Var_S (\hat f_{\it part}(\mathcal{P}))  = \frac{1}{n^2} \sum_{k=1}^r n_k^2 \cdot \Var_S\left[f(x_k)  \right ] \\
            & = \frac{1}{n^2} \sum_{k=1}^r n_k^2 \cdot (\E_S [f(x_k)^2] - \E_S [f(x_k)]^2).
        \end{align*}
        We then use the fact that $f(x_k)^2 = f(x_k)$ and $\E_S[f(x_k)] = \sum_{v_j \in V_k} f(v_j) / n_k$, and take expectation
        when $f$ is drawn from a distribution, to obtain
        \begin{align*}
            & \E_M[\Var_S (\hat f_{\it part}(\mathcal{P})) ] = \frac{1}{n^2} \sum_{k=1}^r \left( (n_k-1) \sum_{v_j \in V_k} \E_M[f(v_j)] \right.\\
            & \quad  \left. - \sum_{v_i,v_j \in V_k, v_i\ne v_j} \E_M [f(v_i)f(v_j)]  \right).
        \end{align*}
        Notice that for any two binary random variables $A$ and $B$, we have $\E[AB]=\frac{1}{2}\left(\Pr[A=B]+\E[A]+\E[B]-1\right)$.
        After applying this formula to $\E_M[f(v_i)f(v_j)]$ and simplifying the expression, we obtain the theorem.
        %	\begin{align*}
        %	&\E_M\left[\Var_S \left(\hat f_{\it part}(\mathcal{P})\right) \right]  
        %	 = \frac{1}{2n^2} \sum_{k=1}^r \sum_{v_i,v_j \in V_k, v_i\ne v_j} (1-\sigma_{ij}).
        %	\end{align*}
    \end{proof}

    %The proof of the theorem requires involved computation and is included in the supplemental material.
    The intuition of the theorem is that, small cost function indicates small volume of each group, 
    which implies that the nodes within each group have high opinion similarities.
    Theorem \ref{theorem:objective function} makes precise our intuition that 
    grouping people with similar opinions 
    would make partitioned sampling more efficient.

    Theorem~\ref{theorem:objective function} provides the connection between 
    the OPS problem and the graph partitioning problem.
    In particular, it suggests that we can reduce the OPS
    problem to the following {\em Min-$r$-Partition} problem: given an undirected graph with non-negative edge weights, partition the graph into $r$ groups such that the sum
    of all groups' volumes is minimized.
    However, Min-$r$-Partition is NP-hard to approximate to within any finite factor \cite{kann1997hardness}, 
    and to the best of our knowledge, there is no approximation or heuristic algorithms in the literature.
    The good news is that Min-$r$-Partition and its dual problem (Max-$r$-Cut) are equivalent in the exact solution, 
    and there exist both approximation and heuristic algorithms for Max-$r$-Cut.
    \citet{maxkcut} 
    propose a semi-definite programming (SDP) algorithm which achieves
    $1-1/r + 2\ln r/r^2$ approximation ratio and is the best to date.
    We adopt the SDP algorithm to solve the OPS problem.
    The SDP partitioning algorithm including the SDP relaxation
    program is given 
    \ifdefined\full 
    in Appendix \ref{section:SDP algorithm}.
    \else
    in the full report \cite{HLC15}.
    \fi
    The drawback of the SDP partitioning algorithm is its inefficiency.
    Thus, we further propose a greedy algorithm to deal with larger graphs, which takes the idea from a heuristic algorithm for Max-$r$-Cut \cite{zhu2013max}.
    
    Given a simple partition $\mathcal{P} = \{(V_1,1), \ldots, (V_r,1)\}$
    and an external node $v_i$ which does not belong to $V_k$ for any $k\in[r]$,
    we define $\delta g_\ell(v_i,\mathcal{P})$ to be $g(\mathcal{P'})-g(\mathcal{P})$,
    where $\mathcal{P'}$ %is the simple partition of $V\cup \{v_i\}$ and
    is $\{(V_1,1), \ldots, (V_\ell\cup\{v_i\},1), \ldots, (V_r,1)\}$.
    Thus $\delta g_\ell(v_i,\mathcal{P})$ represents the increase of the cost function when the external node $v_i$ is added to the group $V_\ell$ of $\mathcal{P}$.
    The greedy algorithm (Algorithm \ref{al:greedy}) first assigns each ungrouped node $x_i$ to the group 
    such that the objective function $g(\mathcal{P})$ is increased the least.
    After the first round of greedy assignment, the assignment procedure is repeated to further
    decrease the cost function, until some stopping condition holds, such as 
    the decrease is smaller than a predetermined threshold.
    
    \begin{algorithm}[t]
        \caption{Greedy Partitioning Algorithm}\label{al:greedy}
        \begin{algorithmic}[1]
            \Require Graph $G_a$ with $n$ nodes, number of groups $r$. 
            \State Randomly generate a node sequence of all the nodes: $x_1$, $x_2$, \dots, $x_n$.
            \State Let $V_1=\ldots =V_r = \emptyset$. 
            \Repeat
            \For {$i\gets 1$ \textbf{to} $n$}
            \If{$x_i \in V_j$ for some $j\in [r]$}{ $V_j=V_j\setminus \{x_i\}$.}
            \EndIf
            \State $k \leftarrow \arg\min_{\ell\in [r]} \delta g_\ell(x_i, \{(V_1,1),\ldots, (V_r, 1)\})$
            \State $V_k \leftarrow V_k \cup \{x_i\}$.
            \EndFor
            \Until{a predetermined stopping condition holds.}
            \State \textbf{Output:}  Partition $\mathcal{P} = \{(V_1,1),\ldots, (V_r, 1)\}$.
        \end{algorithmic}
    \end{algorithm}
    
    The running time of one-round greedy assignment is $O(n+m)$ where $m$ is the number of edges in $G_a$. 
    In our experiment, we will show that greedy partitioning %with a reasonable stopping condition
    performs as well as SDP partitioning but could run on much larger graphs.
    Theoretically, the performance of partitioned sampling using the simple partition generated by the greedy partitioning algorithm
    is always at least as good as naive sampling,
    even using the partition generated after the first round of greedy assignment, 
    as summarized below:
    
    \begin{restatable}{lemma}{greedy}\label{lemma:greedy vs naive}
        Given a vertex set $V$ with % $n$ nodes and 
        sample size $r$, 
        partitioned sampling using the simple partition $\mathcal{P}$
        generated by the greedy partitioning algorithm (even after the first round)
        is at least as good as naive sampling. Specifically,
        \[
        \E_M[\Var_S(\hat f_{\it part}({\cal P}))] \leq 
        \E_M[\Var_S(\hat f_{\it naive}(V,r))].
        \]
    \end{restatable}
    
    We call a partition $\mathcal{P}'$ a {\em refined} partition of $\mathcal{P}$, if each group
    of ${\cal P}'$ is a subset of some group of $\mathcal{P}$.
    Suppose we are given a partition $\mathcal{P}$ such that there exists some group which is allocated more than one sample. 
    Then we can further partition that group by the greedy partitioning algorithm
    and finally obtain a refined simple partition of $\mathcal{P}$.
    According to Lemma \ref{lemma:greedy vs naive}, 
    the refined simple partition should be at least as good as the original partition $\mathcal{P}$,
    summarized as below:
    
    \begin{restatable}{theorem}{simple}\label{theorem: simple is better}
        For any non-simple partition $\mathcal{P}$, there exists a refined simple partition $\mathcal{P}'$ of $\mathcal{P}$, 
        which can be constructed efficiently,
        such that partitioned sampling using the refined simple partition $\cal P'$ is 
        at least as good as partitioned sampling using the original partition $\cal P$. Specifically,
        \[
        \E_M[\Var_S(\hat f_{\it part}({\cal P}'))] \leq 
        \E_M[\Var_S(\hat f_{\it part}({\cal P}))].
        \]
    \end{restatable}
    
    Theorem \ref{theorem: simple is better} shows the superiority of simple partitions, and justifies that it is enough for us to only optimize for partitioned
    sampling with simple partitions.
    
    \subsection{Dealing with Inaccurate Similarities}\label{section: discussion of similarity}
    When accurate opinion similarities are not available, 
    one still can use a {\em balanced partition}
    (i.e., all groups have the exact same size) to achieve as least good sampling result as naive sampling, % for any similarity input,
    summarized as below:
    
    \begin{restatable}{theorem}{balanced}\label{thm:balancedpartition}
        Given a vertex set $V$ with $n$ nodes and sample size $r$ where $n$ is a multiple of $r$, partitioned sampling using any balanced simple partition $\mathcal{P}$ %of $V$
        is at least as good as naive sampling. That is,
        $\Var_S(\hat f_{\it part}({\cal P})) \leq 
        \Var_S(\hat f_{\it naive}(V,r))$ holds for any fixed opinions $f(v_1)$, \dots, $f(v_n)$.
    \end{restatable}
    
    Theorem \ref{thm:balancedpartition} provides a safety net showing that partitioned sampling would not hurt sampling quality. 
    Thus, we can always use the greedy algorithm with a balance partition constraint to achieve better sampling result.
    The result will be further improved if opinion similarities get more accurate.
    
    Furthermore, in the experiment on the real-world dataset (Section \ref{section: experimental evaluation}), 
    we artificially remove all the opinion similarity information (set as $0.5$) between disconnected individuals, and perturb the rest opinion similarities more than $30\%$, to simulate the condition of missing and inaccurate similarities. 
    The experimental result shows that the performance of the greedy algorithm with perturbed inputs is quite close to the performance of the greedy algorithm with exact inputs.
    This demonstrates the robustness of our greedy algorithm in the face of missing and inaccurate opinion similarity data.
    
    Moreover, since real-world social interaction can be characterized well by opinion evolution models, 
    we adapt a well-known opinion evolution model and give an exact computation of opinion similarity based on the model in the next section.
    The model essentially provides a more compact representation than pairwise similarities.
    
    \section{Opinion Evolution Model}\label{section: opinion evolution model and opinion similarity}
    
    We adapt the well-known {\em voter model} to describe social dynamics \cite{VoterModel,DiscreteOpinion}.
    Consider a weighted directed social graph $G=(V, A)$ where $V=\{v_1, v_2, \dots, v_n\}$ is the vertex set and $A$ is the weighted adjacency matrix. 
    Each node is associated with both an {\em innate opinion} and an {\em expressed opinion}. 
    The innate opinion remains unchanged from external influences, while the expressed opinion could be shaped by the opinions of one's
    neighbors, and is the one observed by sampling.
    At initial time, 
    each node $v_i$ generates its innate opinion $f^{(0)}(v_i) \in \{0,1\}$ from an i.i.d.\ Bernoulli distribution with expected value $\mu^{(0)}$.
    The use of i.i.d.\ distribution for the innate opinion is due to the lack of prior knowledge on a brand-new 
    topic, and is also adopted in other models \cite{SocialSampling}.
    When $t>0$, each node $v_i$ updates its expressed opinion $f^{(t)}(v_i) \in \{0,1\}$ independently according to a Poisson process with updating rate $\lambda_i$: 
    at its Poisson arrival time $t$, node $v_i$ sets $f^{(t)}(v_i)$ to its 
    innate opinion with an {\em inward probability} $p_i>0$, or with probability $(1-p_i)A_{ij}/\sum_{k=1}^{n} A_{ik}$, adopts its out-neighbor $v_j$'s expressed opinion
    $f^{(t)}(v_j)$.
    We call the model  {\em Voter model with Innate Opinions (\model{})}.
    
    The \model{} model reaches a steady state if the joint distribution 
    of all node's expressed opinions no longer changes over time.\footnote{The \model{} model has a unique joint
        distribution for the final expressed 
        \ifdefined\full
        opinions. See Appendix \ref{section: steady state} for the proof.
        \else
        opinions \cite{HLC15}.
        \fi
    } 
    We use notation $f^{(\infty)}(v_i)$ to represent the steady-state expressed opinion of node $v_i$, which is a random variable.
    We assume that opinion sampling is done in the steady state, which means that people have sufficiently communicated within the social network.
    
    To facilitate analysis of the VIO model,
    we take an equivalent view of the \model{} model as {\em coalescing random walks} on an augmented graph $\overline G = (V\cup V', 
    E\cup \{e_1^\prime, e_2^\prime, \dots, e_n^\prime\})$, where $V^\prime=\{v_1', v_2', \dots, v_n'\}$ is a copy of $V$, $E$ is the edge set of $G$ and $e_i'=(v_i, v_i')$ for all $i$. In this viewpoint, we have $n$ walkers randomly wandering on $\overline G$ ``back in time'' as follows. At time $t$, all walkers are separately located at $v_1, v_2, \cdots, v_n$.
    Suppose before time $t$, $v_i$ is the last node who updated its expressed opinion at time $\tau<t$,
    then the $n$ walkers stay stationary on their nodes from time $t$ until time $\tau$ ``back in time''. At time $\tau$, the walker at node $v_i$ takes a walk step:
    she either walks to $v_i$'s out-neighbor $v_j\in V$ with probability $(1-p_i)A_{ij}/\sum_{k=1}^{n} A_{ik}$, 
    or walks to $v_i' \in V'$ with probability $p_i$.
    If any walker (e.g., the walker starting from node $v_i$) walks to a node (e.g., $v_k'$) in $V'$, 
    then she stops her walk.  In the \model{} model language, this is equivalent to saying that $v_i$'s opinion at time $t$ is determined
    by $v_k$'s innate opinion, namely $f^{(t)}(v_i)=f^{(0)}(v_k)$.
    If two random walkers meet at the same node in $V$ at any time, they 
    walk together from now on following the above rules (hence the name {\em coalescing}).
    Finally, at time $t=0$, if the walker is still at some node $v_i\in V$, she always walks to
    $v_i' \in V'$.
    
    We now define some key parameters based on the coalescing random walk model, which will be directly used for computing the opinion similarity later. 
    
    \begin{restatable}{definition}{walkparamdef}\label{def:walkparam}
        %\begin{definition} 
        Let ${\cal I}_{ij}^\ell$ denote the event that two random walkers starting from $v_i$ and $v_j$ at time $t=\infty$ eventually meet and the first node they meet at is $v_\ell\in V$.
        Let $Q$ be the $n\times n$ matrix where $Q_{ij}$ denotes the probability
        that a random walker starting from $v_i$ at time $t=\infty$
        ends at $v_j^\prime \in V'$.
        %\end{definition}
    \end{restatable}	
    
    \begin{restatable}{lemma}{walkparam}\label{lemma:walkparam}
        For $i,j,\ell\in [n]$, $\Pr\left[{\cal I}_{ij}^\ell\right]$ is the unique solution of the following
        linear equation system:
        \begin{equation*}
            \Pr\left[{\cal I}_{ij}^\ell\right]=
            \begin{cases}
                0, & i=j\neq \ell, \\
                1, & i=j=\ell,\\
                \sum_{a=1}^n \frac{\lambda_i (1-p_i) A_{ia} } {(\lambda_i+\lambda_j)d_i} \Pr[{\cal I}_{a j}^\ell] \\\quad
                +\sum_{b=1}^n \frac{\lambda_j (1-p_j) A_{jb} } {(\lambda_i+\lambda_j)d_j} \Pr[{\cal I}_{i b}^\ell]
                , & i\neq j,
            \end{cases}
        \end{equation*}
        where $d_i=\sum_{j=1}^n A_{ij}$ is $v_i$'s weighted out-degree.
        In addition, matrix $Q$ is computed by
        \begin{equation*}
            Q=\left(I-\left(I-P\right)D^{-1}A\right)^{-1}P,
        \end{equation*}
        where  $P=\diag(p_1, \dots, p_n)$ and $D = \diag(d_1, \dots, d_n)$
        are two diagonal matrices, and matrix $I-\left(I-P\right)D^{-1}A$ is invertible
        when $p_i>0$ for all $i\in [n]$.
    \end{restatable}
    
    Our main analytical result concerning the \model{} model is the following exact computation of pairwise opinion correlation, which directly leads to opinion similarity:
    
    \begin{restatable}{lemma}{correlation}\label{lemma:correlation}
        For any $i,j\in [n]$, opinion correlation $\c_{ij}$ in the steady state is 
        equal to the probability that two coalescing random walks starting from $v_i$ and $v_j$ at time $t=\infty$ end at the same absorbing node in $V^\prime$.
        Moreover, opinion correlation $\c_{ij}$ can be computed by
        \begin{align*}
            \c_{ij}&=\Cor_M\left(f^{(\infty)}(v_i),f^{(\infty)}(v_j)\right) \\
            &= \sum_{k=1}^n Q_{ik}Q_{jk} + \sum_{\ell=1}^n\Pr\left[{\cal I}_{ij}^\ell\right]\left(1-\sum_{k=1}^n Q_{\ell k}^2\right)
        \end{align*}
        where ${\cal I}_{ij}^\ell$ and $Q$ are defined in Definition \ref{def:walkparam}, and
        $\Pr\left[{\cal I}_{ij}^\ell\right]$ and $Q$
        are computed by Lemma \ref{lemma:walkparam}.
    \end{restatable}
    
    \begin{restatable}{theorem}{similarity}\label{theorem: similarity}
        For any two nodes $v_i$ and $v_j$, their opinion similarity $\s_{ij}$ in the steady state of the \model{} model is equal to:
        $$\s_{ij}=1-2\mu^{(0)}(1-\mu^{(0)})(1-\c_{ij})$$
        where opinion correlation $\c_{ij}$ is computed by Lemma \ref{lemma:correlation}.
    \end{restatable}
    
    Notice that for partitioning algorithms, we only need $1-\s_{ij}$ as the edge weight and by the above theorem this weight value is proportional to $1 - \rho_{ij}$, which means the exact value of $\mu^{(0)}$ is irrelevant for partitioning algorithms.
    \ifdefined\full
    In Appendix \ref{section: correlation computing proof}, 
    \else
    In the full report \cite{HLC15}, 
    \fi
    we will provide an efficient computation of all pairwise opinion correlations 
    with running time $O(nm R)$ by
    a carefully designed iterative algorithm,
    where $m$ is the number of edges of $G$ which is commonly sparse, and $R$ is the number of iterations. 
    We further remark that the correlations are calculated offline based on the existing network and historical data, 
    and thus the complexity compared to the sampling cost of telephone interview or network survey is relatively small. 
    %Moreover, by Theorem \ref{thm:balancedpartition}, even if we save the offline computation cost by sacrificing the accuracy of correlations, we can still achieve better performance than naive sampling using balanced partitions.
    %\liang{which is better? previous expression for the last sentence:Moreover, according to Theorem \ref{thm:balancedpartition}, even if we cannot accurately estimate all correlations 
    %when the offline computation cost is high, we can still achieve better sampling performance using balanced partition.}
    
    \ifdefined\full
    In Appendix \ref{section: more on vio}, 
    \else
    In the full report \cite{HLC15}, 
    \fi
    we further extend the VIO model to include (a) non-i.i.d. distributions of the innate opinions, and (b) negative edges as in
    the signed voter model \cite{Li0WZ15}.
    
    \section{Experimental Evaluation}\label{section: experimental evaluation} 
    
    In this section, we compare the sampling quality of
    partitioned sampling using greedy partitioning (\textsf{Greedy}) 
    and partitioned sampling using SDP partitioning\footnote{We use CVX package \cite{cvx,gb08} to solve the SDP programming.} (\textsf{SDP}) against 
    naive sampling (\textsf{Naive}) based on the \model{} model, using both synthetic and real-world datasets.
    We describe major parameter settings for the experiments below, while leave the complete settings in
    \ifdefined\full
    Appendix \ref{section: experimental setting}.
    \else
    the full report \cite{HLC15} due to space constraints.
    \fi

    In our experiment, when the parameters of \model{} model are set, the simulation is done by (a) calculating the pairwise opinion similarities by Theorem \ref{theorem: similarity}, (b) running the partitioning algorithms to obtain the partition candidate, and (c) computing the expected variance $\E_M[\Var_S(\hat f)]$ by 
    Theorem~\ref{theorem:objective function}.
    
    \begin{figure}[t] 
        \centering
        \def \mag {0.224}
        \subfigure[Synthetic graph ($100$ nodes)] {      
            \includegraphics[width=\mag\textwidth]{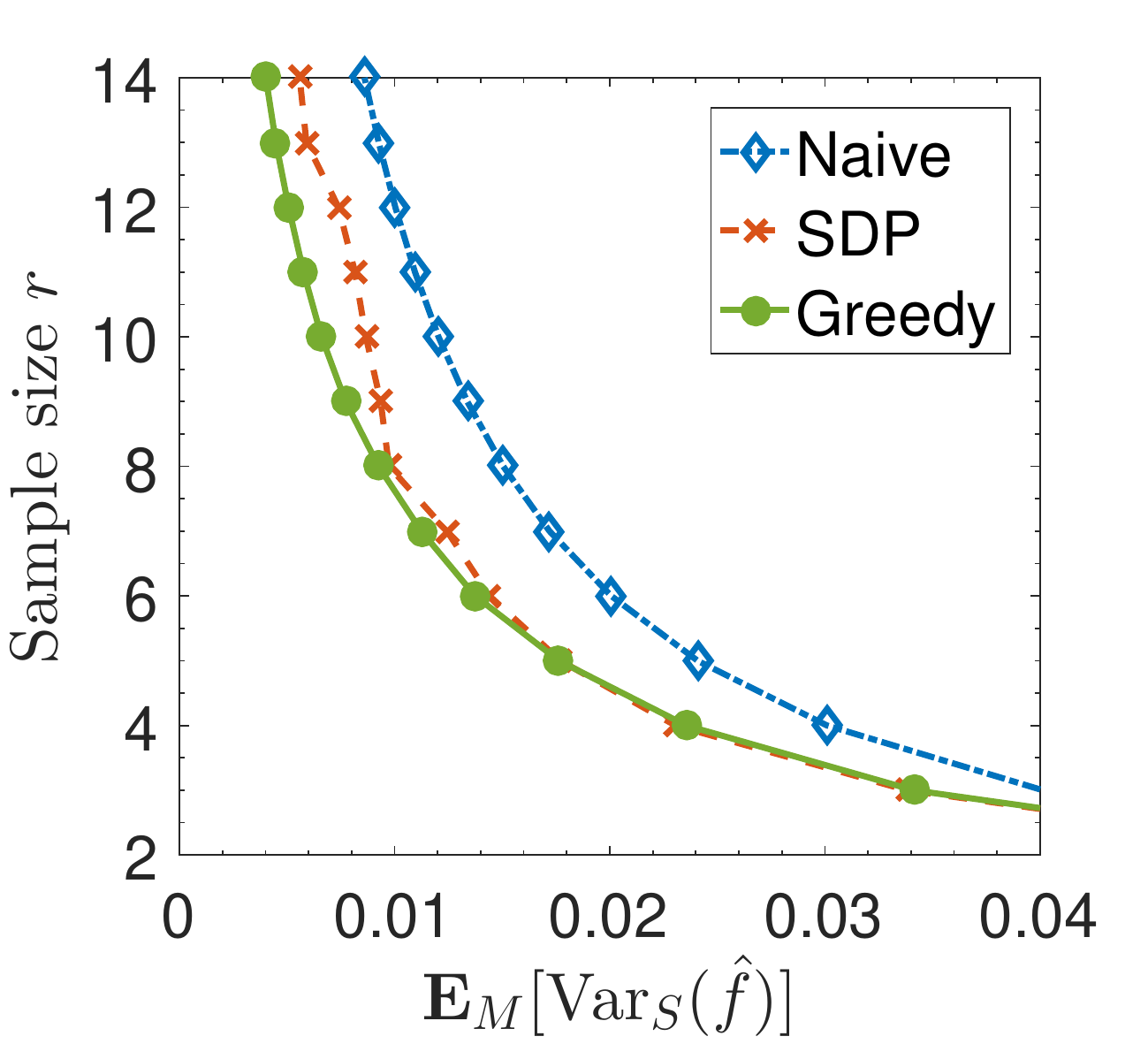}  \label{fig:small}
        }
        \subfigure[Synthetic graph ($\num{10000}$ nodes)] { \label{fig:cluster}     
            \includegraphics[width=\mag\textwidth]{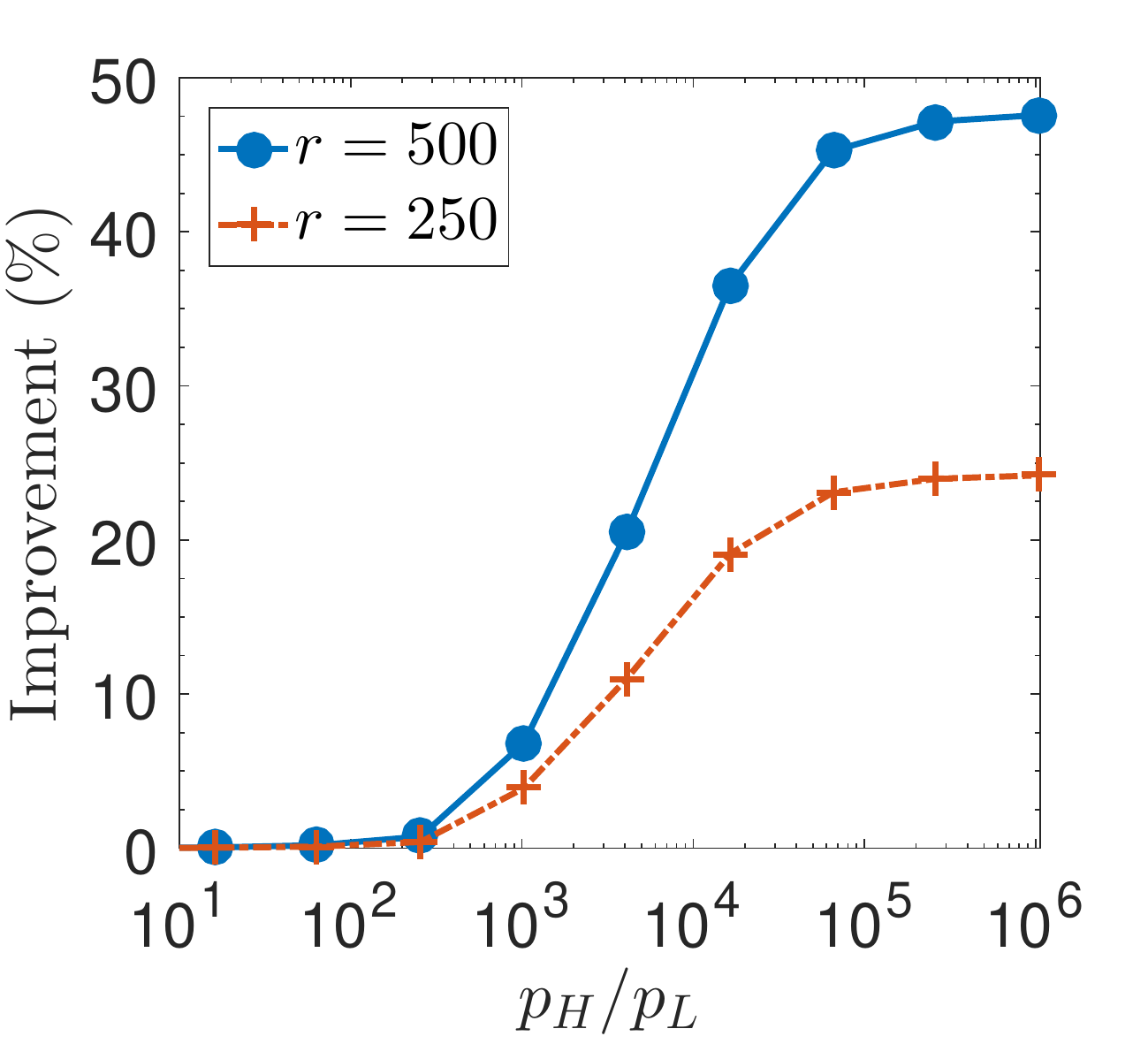} 
        }
        \subfigure[Synthetic graph ($\num{10000}$ nodes)] { \label{fig:inward}     
            \includegraphics[width=\mag\textwidth]{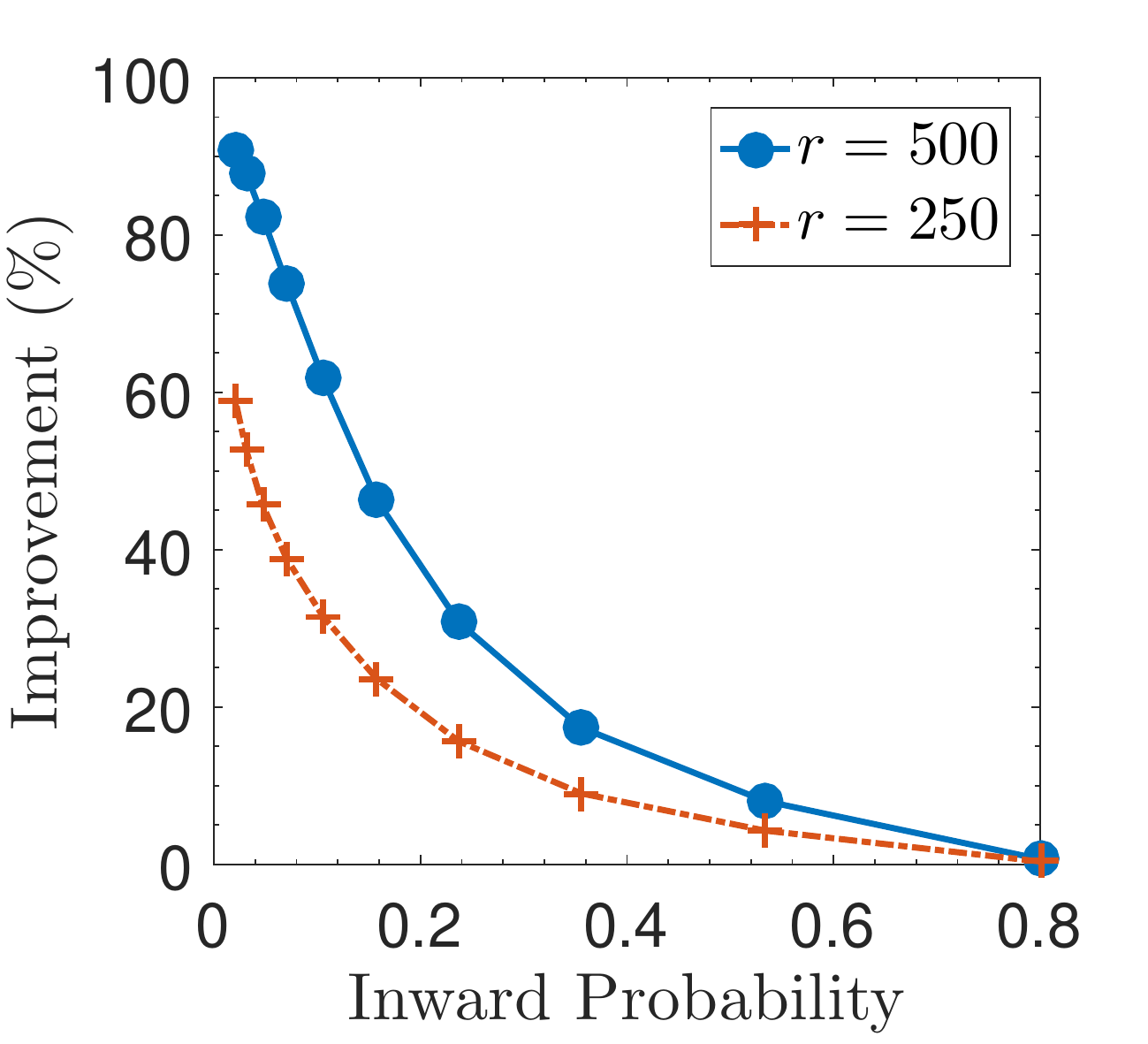}  
        }
        \subfigure[Inward probability distribution] { \label{fig:p distribution}     
            \includegraphics[width=\mag\textwidth]{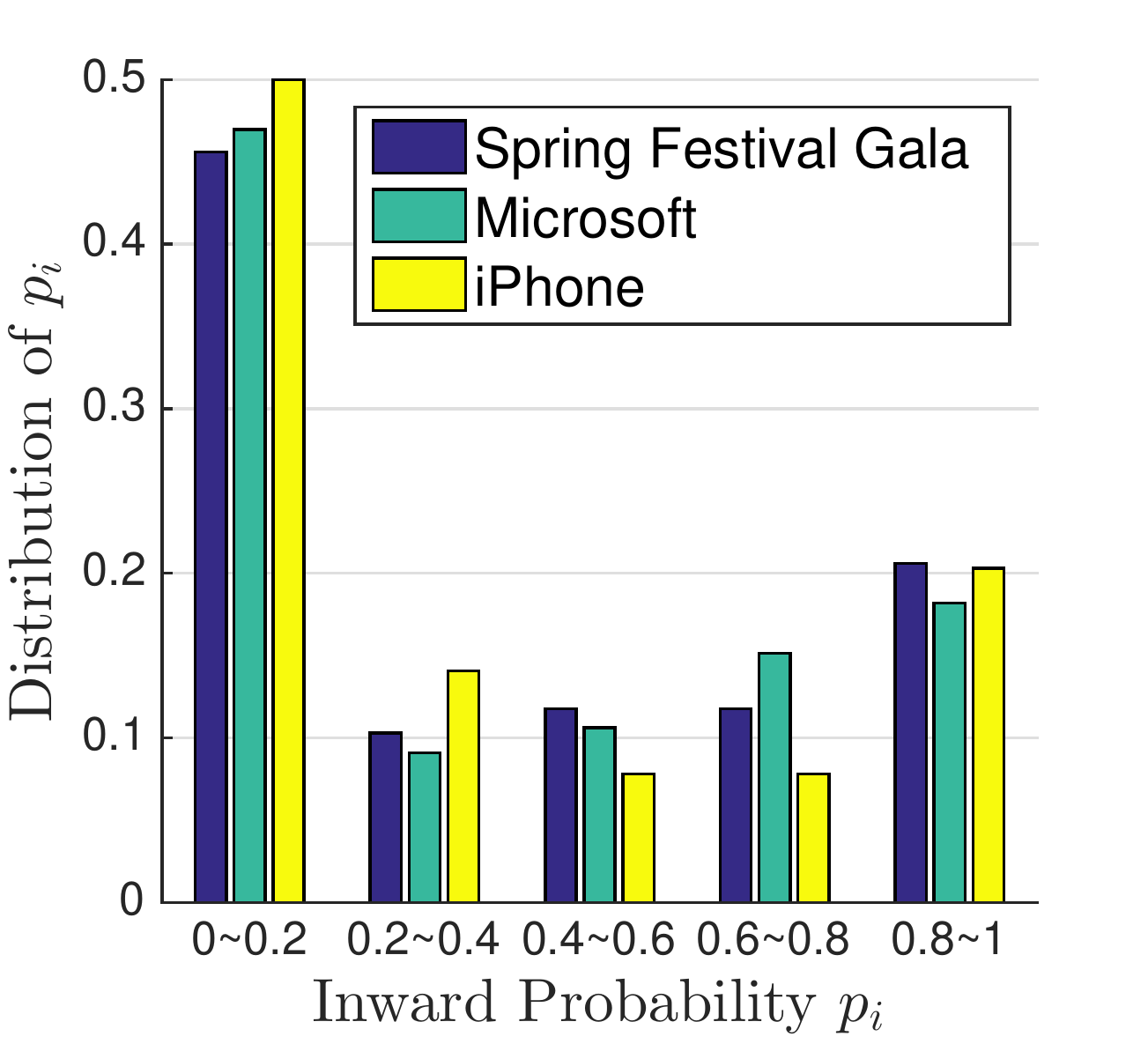}  
        }
        \subfigure[Weibo graph with $\lambda_i=1$] { \label{fig:lambda=1}     
            \includegraphics[width=\mag\textwidth]{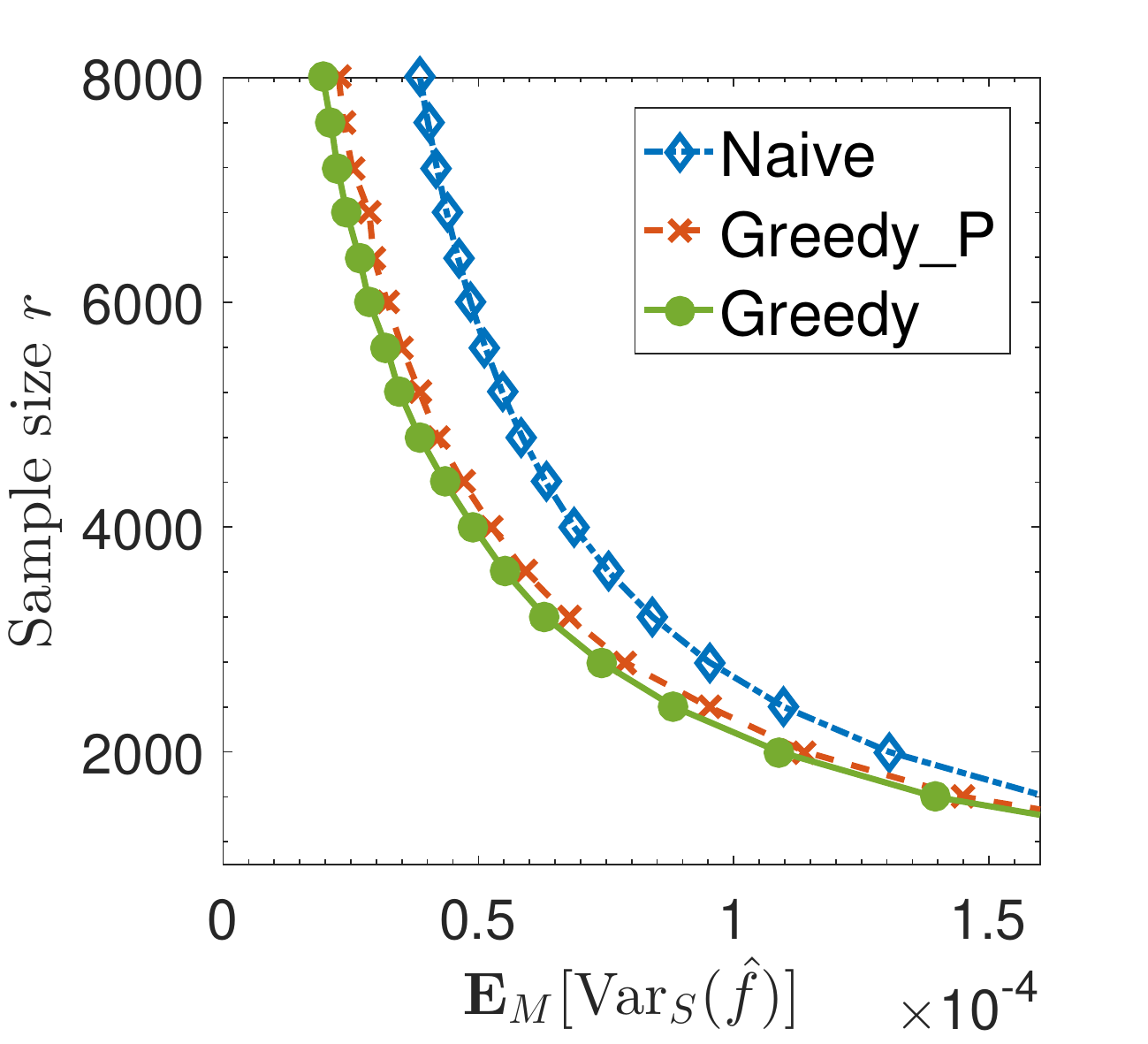} 
        }
        \subfigure[Weibo graph with $\lambda_i$ set to the number of $v_i$'s tweets in a year] { \label{fig:poisson}     
            \includegraphics[width=\mag\textwidth]{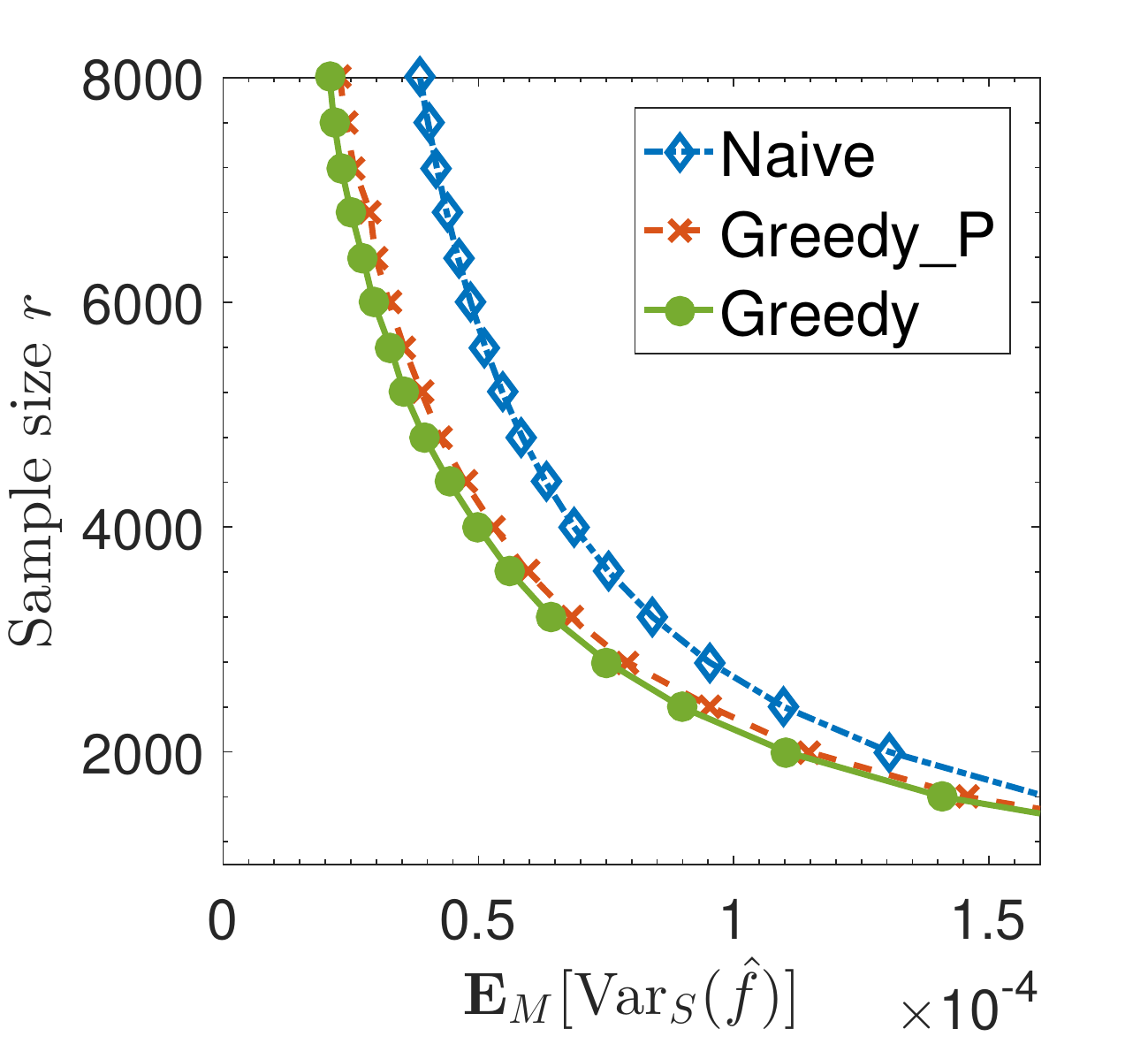}
        }
    \end{figure} 
    
    \textbf{Synthetic Dataset.} We use the planted partition model \cite{HiddenPartition} to generate undirected graphs, which aims at resembling the community structure in real-world social networks. 
    Given $n$ vertices and $k$ latent disjoint groups,
    every edge $(v_i,v_j)$ is generated with a high probability $p_H$ if $v_i$ and $v_j$ are in the same latent 
    group, otherwise with a low probability $p_L$.

    We generate two different sizes of synthetic graphs.
    The small one includes $100$ nodes and $20$ latent groups,
    and $p_H$, $p_L$ and $\lambda_i$ are set to $0.9$, $0.01$ and $1$, respectively.
    The inward probability of each node is randomly chosen from $[0,0.01]$.
    Fig \subref{fig:small} shows that,
    when the sample size $r$ is small, the performance of \textsf{SDP} and \textsf{Greedy} are similar to each other and both better than \textsf{Naive}.
    When the sample size $r$ increases,
    \textsf{Greedy} becomes much better than \textsf{Naive}, and \textsf{SDP} starts getting worse.
    For the large synthetic graph with $10$k nodes and $500$ latent groups, 
    \textsf{SDP} is no longer feasible, thus we compare the
    improvement of \textsf{Greedy} against \textsf{Naive}.
    In Fig \subref{fig:cluster},
    we range $p_H/p_L$
    and find that larger $p_H/p_L$ (more apparent clustering) indicates the better performance of the partitioned sampling method.
    When $p_H/p_L$ increases from $10^3$ to $10^5$, the improvement of expected sample variance increases rapidly.
    When $p_H/p_L>10^5$, 
    the improvement becomes saturated.
    This is because the number of edges which cross different latent groups are so few that it decreases rather slowly and the graph structure is almost unchanged when $p_H/p_L$ increases further.
    In Fig \subref{fig:inward},
    we set all nodes' inward probabilities to be equal and vary them from $0.02$ to $0.8$.
    The figure shows that the lower inward probability leads to the better performance of partitioned sampling.
    When the inward probability gets small, the improvement expected sample variance increases rapidly.
    This is because a lower inward probability means people interacting more with each other
    and thus their opinions are correlated more significantly.
    According to the above experiments, we conclude that
    the larger $p_H/p_L$
    and the lower inward probability 
    make people's opinions more clustered and more correlated inside the clusters, and our partitioned
    sampling method works better for these cases.

    \textbf{Real-World Dataset.}
    We use the micro-blog data from \url{weibo.com} \cite{weibodata},
    which contains $\num{100102}$ users and $\num{30518600}$ tweets within a one-year timeline from 1/1/2013 to 1/1/2014.
    We treat the user following relationship between two users as a directed edge (with weight $1$).
    
    We first learn the distribution of user's inward probabilities from the data.
    We extract a series of users' opinions on $12$ specific topics (e.g., Microsoft, iPhone, etc.) by applying a keyword classifier and a sentiment analyzer \cite{sentiment} to the tweets.
    We also collect their relationships and form a subgraph for each topic.
    Then we use VIO model to fit the data by solving a minimization problem w.r.t. inward probabilities using gradient descent.
    Fig \subref{fig:p distribution} shows the distribution of inward probabilities for 
    three of the topics, namely Spring Festival Gala ($68$ users), Microsoft ($66$ users) and iPhone ($59$ users), and the results for
    other topics are similar.
    %The distributions among different topics are quite similar:
    From these distributions, we observe that
    (a) over $45\%$ inward probabilities locate in $[0, 0.2]$; 
    (b) the probability that $p_i$ locates in $[0.8, 1]$ is the second highest;
    (c) others almost uniformly locate in $[0.2, 0.8]$.
    This indicates that in the real world, most people tend to adopt others' opinions,
    which matches the intuition that people are often affected by others.
    We manually look up the users who locate in $[0.8, 1]$, 
    and find that most of them are media accounts and verified users.
    This matches our intuition that those users always take effort to spread their 
    own opinions on the web but rarely adopt others' opinions,
    hence they should have large inward probabilities.
    
    Now we simulate the sampling methods on the Weibo graph.
    We first remove the users who do not follow anyone iteratively,
    and get the graph including $\num{40787}$ nodes and $\num{165956}$ directed edges.
    We generate each user's inward probability following the distribution we learned.
    We use two different settings for opinion updating rates:
    one is to set $\lambda_i=1$ for all $i\in[n]$;
    the other is to set $\lambda_i$ to the number of $v_i$'s tweets in a year.
    The improvement of \textsf{Greedy} against \textsf{Naive}
    with two different updating rate settings are similar as shown in Fig \subref{fig:lambda=1} and \subref{fig:poisson}.
    In particular, if we fix $\E_M[\Var_S(\hat f)]$ to be $\num{3.86e-5}$, \textsf{Greedy} needs $4794$ samples while \textsf{Naive} needs $8000$ samples (saving $40.1\%$) in Fig~\subref{fig:lambda=1},
    and \textsf{Greedy} needs $4885$ samples while \textsf{Naive} needs $8000$ samples (saving $38.9\%$) in Fig~\subref{fig:poisson}.
    This indicates that partitioned sampling greatly improves the sampling quality,
    and the sample size saving is more apparent when
    the expected sample variance gets smaller (i.e., the requirement of sampling quality gets higher).
    Moreover, in order to test the performance of partitioned sampling with missing and inaccurate opinion similarities,
    we artificially remove all the opinion similarity information between disconnected nodes (set similarities as 0.5), and 
    perturb each rest similarity $\sigma_{ij}$ with a random noise $e_{ij}$ in the range $[-0.1-30\% \cdot \sigma_{ij}, 0.1+30\% \cdot \sigma_{ij}]$
    (set perturbed similarity of $\sigma_{ij}$ as the median of $\{0, \sigma_{ij}+e_{ij} , 1\}$).
    %in the range of $\pm (0.1+30\% \cdot \sigma_{ij}) $.
    Fig \subref{fig:lambda=1} and \subref{fig:poisson} show that \textsf{Greedy} using the above perturbed similarities (denoted as \textsf{Greedy\_P}) is very close to \textsf{Greedy}, and still has a 
    significant improvement against naive sampling.
    
    In conclusion, the experimental results demonstrate the excellent performance of our partitioned sampling 
    method both on synthetic and real-world datasets, even when the opinion similarities are missing or inaccurate.
    %\clearpage
    \bibliography{ref}
    \bibliographystyle{aaai}
    
    \ifdefined\full\else
\end{document}
\fi

\clearpage
{\parindent0pt
    \onecolumn
    \appendix
    \section*{{\LARGE Appendix}}
    \mbox{}\\
    The appendix is organized as follows.
    In Section \ref{section:object discussion}, we show that using variance $\Var_{M,S}(\hat f)$ as objective is equivalent to using $\E_M[\Var_S(\hat{f})]$ for the OPS problem .
    In Section \ref{section:SDP algorithm}, we present the formulation of our SDP partitioning algorithm.
    In Section \ref{section: steady state}, we show that the \model{} model has a unique steady state.
    In Section \ref{section: correlation computing proof}, we provide an efficient computation of pairwise opinion correlations.
    In Section \ref{section: experimental setting}, we provide the implementation details of the experiment.
    In Section \ref{section: proof}, we provide all the mathematical proofs of lemmas and theorems stated in the main paper.
    In Section \ref{section: more on vio}, we provide more discussion about the VIO model, including two extensions of the VIO model.

    \section{Objective Function of the OPS Problem}\label{section:object discussion}
In the definition of the OPS problem (Definition~\ref{def:OPS}),
    we use the expected variance $\E_M[\Var_S(\hat f)]$ as the objective function.
Another intuitive setting of the objective function is $\Var_{M,S}(\hat f)$
    which combines all the randomness into the variance together.
We now show that these two objective functions are equivalent.
Actually,
\begin{align*}
    \Var_{M,S}(\hat f) &= \E_{M,S}[ \hat f^2] - \E_{M,S}[ \hat f ]^2 
    = \E_M\E_S[\hat f^2] - (\E_M\E_S[\hat f])^2.
\end{align*}
Due to the unbiasedness of $\hat f$, we have
$\E_S[\hat f]=\bar f$.
Thus $$\Var_S(\hat f)=\E_S[\hat f^2]-\E_S[\hat f]^2=\E_S[\hat f^2]-\bar f^2.$$
Therefore 
\begin{align*}
    \Var_{M,S}(\hat f) &= \E_M\E_S[\hat f^2] - (\E_M\E_S[\hat f])^2\\
    &= \E_M\left[ \Var_S(\hat f) +\bar f^2\right] - \left(\E_M[\bar f]\right)^2\\
    &= \E_M[\Var_S(\hat f)] + \Var_M(\bar f).
\end{align*}
Since $\Var_M(\bar f)$ stays constant with different partitions, for any partition, the different between $\Var_{M,S}(\hat f)$ and $\E_M[\Var_S(\hat{f})]$ is the same.
Thus using variance $\Var_{M,S}(\hat f)$ as objective is equivalent to using $\E_M[\Var_S(\hat{f})]$ for the OPS problem.

    \section{SDP Partitioning Algorithm}\label{section:SDP algorithm}
In this section, we present the formulation of our SDP partitioning algorithm. 
The idea is to partition the assistant graph $G_a$ into $r$ groups by solving the Max-$r$-Cut problem for $G_a$.
The task is to find $r$ groups $V_1, V_2, \dots, V_r$ in order to maximize the following function:
$$\sum_{k\neq l}\Cut_{G_a}(V_k, V_l)$$
where $\Cut_{G_a}(V_k, V_l)$ is defined by $\sum_{v_i\in V_k, v_j\in V_l} w_{ij}$ and $w_{ij}=1-\sigma_{ij}$ is the weight of edge $(v_i,v_j)$ of $G_a$. 

%\citet{maxkcut} propose an approximation algorithm for Max-$r$-Cut using Semi-Definite Programing (SDP) as a relaxation.
%We adopt their algorithm for solving the above problem as follows.

Take an equilateral simplex in $\mathbb{R}^{r-1}$ with vertices $\vec{b}_1$, $\vec{b}_2$, \dots, $\vec{b}_r$. Let $\vec{c}=(\vec b_1+\vec b_2+\dots+\vec b_r)/r$, and let $\vec a_k=\frac{\vec b_k-\vec c}{\left\lVert\vec b_k-\vec c\right\rVert}$ for $1\leq k \leq r$.
Thus $\vec{a}_1$, $\vec{a}_2$, \dots, $\vec{a}_r$ have the following property:
\begin{equation*}
	\vec a_k \cdot \vec a_l=
	\begin{cases}
		1, &\text{if } \vec a_k = \vec a_l; \\
		-\frac{1}{r-1}, & \text{if } \vec a_k \neq \vec a_l.
	\end{cases}
\end{equation*}
We use $\vec y_i\in \{\vec{a}_1$, $\vec{a}_2$, \dots, $\vec{a}_r\}$ to represent which group node $v_i$ is located in. 
If node $v_i$ is in $k$-th group, then $\vec y_i=\vec a_k$. 
In this way, the maximization problem can be written as
\begin{align*}\tag{IP}
	\text{Maximize \quad} &\frac{r-1}{r} \sum_{i\neq j}w_{ij}\left(1-\vec y_i\cdot \vec y_j\right) \\
	\text{Subject to \quad} &\vec y_i \in \{\vec{a}_1, \vec{a}_2, \dots, \vec{a}_r\}, i\in\{1, 2, \dots, n\}.
\end{align*}
We then relax the above optimization problem by replacing $\vec y_i\cdot \vec y_j$ with the $(i,j)$ entry of a positive semi-definite symmetric matrix $Y$ whose diagonal elements are equal to $1$, and relaxing $\vec y_i\cdot \vec y_j$ to be not less than $-\frac{1}{r-1}$.
\begin{align*}%\refstepcounter{equation}\tag{SDP}\label{a:SDP}
	\text{Maximize \quad} &\frac{r-1}{r} \sum_{i\neq j}w_{ij}\left(1-Y_{ij}\right) \\
	\text{Subject to \quad} &Y_{ii}=1, \forall i, \\
	&Y\succeq 0, \\
	&Y_{ij}\geq -\frac{1}{r-1}, \forall i\neq j,\\
	& Y \text{ is symmetric}.
\end{align*}
Our SDP partitioning algorithm is performed by solving the above SDP problem and rounding the SDP-relaxed solution to IP-flexible solution, which is shown in Algorithm~\ref{al:SDP partition}.

\begin{algorithm}[t]
    \caption{SDP Partitioning Algorithm}\label{al:SDP partition}
    \begin{algorithmic}[1]
        \Require Graph $G_a$ with $n$ nodes, number of groups $r$. 
        \State Solve the following SDP problem and compute the Cholesky decomposition of $Y$. Let $y_1$, $y_2$, \dots, $y_n$ be the resulting vectors.
        \Statex \vspace{-3mm}
        {\small\begin{align*}
            %\refstepcounter{equation}
            %\tag{SDP}\label{SDP problem}
            \text{Maximize \quad} & \frac{r-1}{r}\sum_{i\neq j}w_{ij}\left(1-Y_{ij}\right) \\
            \text{Subject to \quad}  &\text{(a) } Y_{ii}=1, \forall i; \text{ (b) } Y_{ij}\geq -\frac{1}{r-1}, \forall i\neq j; 
            \text{ (c) } Y\succeq 0; \text{ (d) } Y \text{ is symmetric}.
            \end{align*}}
        \State Choose $r$ random vectors $z_1$, $z_2$, \dots, $z_r$ from $\mathbb{R}^n$.
        \State Partition $V$ into $r$ groups $V_1$, \dots, $V_r$ according to which of  $z_1$, $z_2$, \dots, $z_r$ is closest to each $y_k$.\footnotemark
        \State \textbf{Output:} Partition $\mathcal{P}= \{(V_1,1),\ldots, (V_r, 1)\}$.
    \end{algorithmic}
\end{algorithm}
\footnotetext{If the partitioning result is less than $r$ groups, we reselect $r$ new random vectors from $\mathbb{R}^n$ and repeat the step again.\nopagebreak}
    \section{Steady State of the VIO Model}\label{section: steady state}
In this section, we show that the \model{} model has a unique steady state, as summarized as follow.
\begin{proposition*}\label{lemma: steady state}
    When $p_i>0$ for all $i\in [n]$, the \model{} model has a unique joint
    distribution for the final expressed opinions in the steady state.
\end{proposition*}

\begin{proof}
    The opinion evolution can be viewed  as a Markov chain.
    Each possible assignment of $v_1, v_2, \ldots, v_n$'s expressed opinions forms one state and the initial state of the Markov chain is $(f^{(0)}(v_1), f^{(0)}(v_2), \ldots, f^{(0)}(v_n))$.
    At each Poisson arrival time, the transition from one state to another represents the change of the opinion assignment. 
    Thus the state space consists of all the states reachable from the initial state.
    The \model{} model has a unique steady state distribution for the final expressed opinions if and only if the Markov chain has a unique stationary distribution.
    In order to prove the existence of the unique stationary distribution of the Markov chain, we only need to prove that the Markov chain is irreducible and aperiodic \cite{MarkovChains}.
    Notice that each state in the state space can be reached from the initial state.
    Meanwhile, each state in the state space can return to the initial state by 
    all the nodes updating their expressed opinions to the innate opinion one by one, which happens
    with a positive probability.
    This means that any two states in the state space are connected, indicating the irreducibility of the Markov chain.
    In addition, the initial state is aperiodic since it has a self-loop in the state transition graph (with probability at least $\sum_{i=1}^n p_i/n>0$).
    An irreducible Markov chain is aperiodic if there exists one aperiodic state.
    Therefore, the Markov chain is irreducible and aperiodic, with the unique stationary distribution being reached after long enough time.
\end{proof}
    \section{Efficient Computation of Opinion Correlation}\label{section: correlation computing proof}

Naive correlation computation directly using Lemma \ref{lemma:walkparam} and \ref{lemma:correlation} 
by solving the linear equation system for $\{\Pr\left[{\cal I}_{ij}^\ell\right]\}$'s would 
have a running time of $O(n^7)$ (See in proof of Lemma~\ref{lemma:walkparam}).
%Instead, we can do iterative computation on $\{\Pr\left[{\cal I}_{ij}^\ell\right]\}$ to reduce
%the running time to $O(n^4 R)$, where $R$ is the number of iterations.
We now improve the running time to $O(nm R)$ by
a carefully designed iterative computation method,
where $m$ is the number of edges of the social graph $G$ and $R$ is the number of iterations.

We first consider the event that a walker starting at $v_i$ is absorbed by $v_j'\in V'$ after $k$ random walk steps in the coalescing random walk model.
We use $q_{ij}^{(k)}$ to represent the probability of the above event happening.
Initially, the walker is located at $v_i$.
With probability $p_i$, she takes her first step to the sink node $v_i'$.
Then the event that she is absorbed by $v_j'\in V'$ after $k$ steps happens if and only if $i=j$.
If she does not walk to the sink node in her first step, but walk to one of her neighbors $v_a$ (happening with probability $(1-p_i)A_{ia}/d_i$),
then the event happens if and only if she walks from $v_a$ to $v_j'$ in her rest $k-1$ steps.
Thus $q_{ij}^{(k)}$ can be computed iteratively by
\begin{align*}
q^{(k)}_{ij}=p_i \cdot 1_{i=j} +\sum_{a=1}^n\frac{(1-p_i)A_{ia}}{d_i}q^{(k-1)}_{aj}.
\end{align*}

The running time of computing all $\{q^{(k)}_{ij}\}$'s from $\{q^{(k-1)}_{ij}\}$'s is
$$\sum_{i=1}^n \sum_{j=1}^n (1+d_i)=O(nm),$$
where $m$ is number of edges of the social graph $G$.
We remark that when $k\rightarrow \infty$, $q_{ij}^{(k)}$ approaches to the $(i,j)$ entry of matrix $Q$ defined in Definition \ref{def:walkparam}. 
Thus $Q$ can be computed in the above iterative way.
We further remark that $Q$ can be computed column by column (fixing subscript $j$) to save running space.

Now we consider two walkers take coalescing random walks on the graph $\overline{G}$ starting at $v_i$ and $v_j$ respectively.
We use notation $M^{(k)}_{ij}$ to denote the probability that their walks coalesce after they taking altogether $k$ random walk steps .
If $i=j$, two walkers have coalesced since the beginning, thus $M^{(k)}_{ij}$ is alway equal to one.
When $i\neq j$, with probability $\frac{\lambda_i}{\lambda_i+\lambda_j}$, the first walk step is taken by the walker starting at $v_i$. 
If she walks to her sink node $v_i'$ (happening with probability $p_i$), then the other walker who is at $v_j$ must walk to the same sink node $v_i'$ in $k-1$ steps (happening with probability $q_{ji}^{(k-1)}$).
If she does not walk to her sink node but one of her neighbors $v_a$ (happening with probability $(1-p_i)A_{ia}/d_i$),
then two walkers will coalesce in the rest $k-1$ steps with probability $M^{(k-1)}_{aj}$.
The case that the first step is taken by the walker starting at $v_j$ follows the similar analysis.
Thus $M^{(k)}_{ij}$ can be calculated by
\begin{align*}
M^{(k)}_{ij}=&\frac{\lambda_i}{\lambda_i+\lambda_j}\left[p_i q_{ji}^{(k-1)}+\sum_{a=1}^n\frac{(1-p_i)A_{ia}}{d_i}M^{(k-1)}_{aj}\right]
+ \frac{\lambda_j}{\lambda_i+\lambda_j}\left[p_j q_{ij}^{(k-1)}+\sum_{a=1}^n\frac{(1-p_j)A_{ja}}{d_j}M^{(k-1)}_{ai}\right].
\end{align*}

The running time of computing $\{M^{(k)}_{ij}\}$'s with one iteration is 
$$\sum_{i=1}^n \sum_{j=1}^n (1+d_i)+(1+d_j)=O(nm),$$
where $m$ is number of edges of the social graph $G$.

According to Lemma \ref{lemma:correlation}, opinion correlation $\c_{ij}$ in the steady state is 
equal to $\lim_{k\rightarrow \infty}M^{(k)}_{ij}$.
Thus we can obtain people's opinion correlations by computing $\{M^{(k)}_{ij}\}$'s and $\{q^{(k)}_{ij}\}$'s iteratively in time $O(nmR)$ where $R$ is the number of iterations.
We remark that for any $i,j\in[n]$, $M^{(k)}_{ij}$ and $q^{(k)}_{ij}$ monotonically increase with increasing $k$,
    and both have the upper bound $1$. 
Thus the above iterative procedure will converge.
    \section{Experimental Setting Details}\label{section: experimental setting} 

In our experiment, when the parameters of \model{} model (i.e., weighted adjacency matrix $A$, people's inward probabilities $p_1$, $p_2$, \dots, $p_n$, updating rates of people's opinions $\lambda_1$, $\lambda_2$, \dots, $\lambda_n$, and the expected value of innate opinion $\mu^{(0)}$) are set, the experiment is done by (a) calculating the pairwise opinion similarities by Theorem \ref{theorem: similarity} and the efficient opinion correlation computation method given in Section \ref{section: correlation computing proof}, (b) running the partitioning algorithms\footnote{The iteration of greedy partitioning algorithm stops when the decrease of the cost function in one iteration is less than 0.01\% of the cost function.} to obtain the partition candidate, and (c) computing the expected variance $\E_M[\Var_S(\hat f)]$\footnote{Each randomized partitioning algorithm was run $10$ times, and we took the average of the expected variance as the result.} by Theorem \ref{theorem:objective function}.
In both synthetic and real-world datasets, we set $\mu^{(0)}$ to be $0.5$.
Notice that the value of $\mu^{(0)}$ has no effect on the results (Theorem \ref{theorem: similarity}).

\subsection{Synthetic Dataset}\label{section:synthetic}

%\begin{figure*}[t] 
%    \centering
%    
%    \subfigure[Synthetic graph ($100$ nodes)] { \label{afig:small}  
%        \includegraphics[width=0.31\textwidth]{node100.eps}  
%    }
%    \subfigure[Synthetic graph ($10,000$ nodes)] { \label{afig:cluster}     
%        \includegraphics[width=0.31\textwidth]{cluster.eps} 
%    }
%    \subfigure[Synthetic graph ($10,000$ nodes)] { \label{afig:inward}     
%        \includegraphics[width=0.31\textwidth]{inward.eps}  
%    }
%    \caption{Experimental results on synthetic graphs.}
%\end{figure*} 

In our synthetic experiments, we use the planted partition model to generate undirected graphs.
It is specified by four parameters: the number of vertices $n$, the number of latent groups $k$, the intra-partition and inter-partition edge probabilities $p_H$ and $p_L$, respectively. 
First, we assign each node to one of the $k$ latent groups uniformly at random.
Next, we independently connect each pair of nodes in the same latent group with probability $p_H$, and two nodes in different latent groups with probability $p_L<p_H$.
For synthetic graphs, we set opinion updating rate $\lambda_i$ to $1$ for all $i\in[n]$, and set the weight of each edge to $1$.
We generate two different sizes of synthetic graphs.

\subsubsection{Small synthetic graph} 
The small synthetic graph we generate includes $100$ nodes and $20$ latent groups.
Probability $p_H$ and $p_L$ are set to $0.9$ and $0.01$, respectively.
The inward probability of each node is randomly chosen from $[0,0.01]$.
We put sample size on $y$-axis in Fig \subref{fig:small} to make it easier to see the savings on the sample size under the same expected sample variance.

\subsubsection{Large synthetic graphs}
The large synthetic graphs includes $10$k nodes and $500$ latent groups.
We run \textsf{Greedy} and \textsf{Naive} using different sample sizes ($r=250$ and $r=500$), varying the inward probabilities and $p_H/p_L$, to observe the improvement of expected sample variance under different graph clustering and inward tendency settings.
In Fig \subref{fig:cluster},
we set all nodes' inward probabilities to $0.05$ and $p_H$ to $1$, and range $p_L$ from $10^{-1}$ to $10^{-6}$.
The improvement on $y$-axis means the improvement of expected sample variance from {\sf Naive}
to {\sf Greedy}.
In Fig \subref{fig:inward},
we set $p_H$ to $1$ and $p_L$ to be $10^{-5}$ to generate the graph.
For this graph, we set all nodes' inward probabilities to be identical, varying from $0.02$ to $0.8$.

\subsection{Real-World Dataset}\label{section:realdata}

%\begin{figure*}[t] 
%    \centering
%    %\setcounter{subfigure}{0}
%    \subfigure[Distribution of inward probabilities] { \label{afig:p distribution}     
%        \includegraphics[width=0.31\textwidth]{p_distribution.eps}  
%    }
%    \subfigure[Weibo graph with $\lambda_i=1$] { \label{afig:lambda=1}     
%        \includegraphics[width=0.31\textwidth]{nonpoisson.eps} 
%    }
%    \subfigure[Weibo graph with $\lambda_i$ set to the number of $v_i$'s tweets] { \label{afig:poisson}     
%        \includegraphics[width=0.31\textwidth]{poisson.eps}
%    }
%    \caption{Experimental results on real-world graphs.}
%\end{figure*} 

The real-network dataset we use is the micro-blog data from \url{weibo.com},
which contains $\num{100102}$ users and $\num{30518600}$ tweets within a one-year timeline from 1/1/2013 to 1/1/2014.
We treat the user following relationship between two users as a directed edge (with weight $1$).
For this dataset, we first need to learn
the distribution of people's inward probabilities.

\subsubsection{Distribution of inward probabilities}\label{section:inward}

In order to observe the evolution of opinions for a specific topic of interest,
We manually choose $12$ specific topics (e.g., Microsoft, iPhone, etc.),
and extract all tweets from the Weibo dataset related to these topics (simply using keyword based classifier).
We then run each tweet through a sentiment analyzer to obtain binary opinion values (positive/negative).
Thus we get a series of opinions for each user at discrete time corresponding to each topic.
For each topic, we select those users who published opinions at least $4$ times,
and regard their first opinions as their innate opinions $f^{(0)}(v_1)$, $f^{(0)}(v_2)$, \dots, $f^{(0)}(v_n)$ and treat the average of the rest opinions as their expected opinions in the steady state state, denoted as $\mu_1$, $\mu_2$, \dots, $\mu_n$.
We then collect their relationships and form a subgraph for the corresponding topic.

Recall the definition of matrix $Q$ (Definition \ref{def:walkparam}), 
and it is easy to see that $\mu_i=\E_M[f^{(\infty)}(v_i)]=\sum_{j=1}^{n}Q_{ij}f^{(0)}(v_j)$.
Then
\begin{equation*}
    \begin{pmatrix}
        \mu_1\\
        \mu_2\\
        \dots \\
        \mu_n
    \end{pmatrix}
    =Q
    \begin{pmatrix}
        f^{(0)}(v_1)\\
        f^{(0)}(v_2)\\
        \dots \\
        f^{(0)}(v_n)
    \end{pmatrix}
    \text{ (or }
    \vec \mu=Q\vec{f}^{(0)}
    \text{)}.
\end{equation*}
Thus we can estimate the inward probabilities by solving the following programming
\begin{align*}
    &\text{Minimize \ }  \left\|\vec{\mu}-Q\vec{f}^{(0)}\right\|,\\
    &\text{Subject to \ }  0\leq p_i \leq 1, \forall i \in [n],
\end{align*}
and we use gradient descent method to handle the above programming.
Fig \subref{fig:p distribution} shows the distribution of inward probabilities for three of the topics.

\subsubsection{Partitioned sampling on Weibo graph}

For the original Weibo dataset, we first remove the users who do not follow anyone, iteratively.
Then we get our Weibo graph including $\num{40787}$ nodes and $\num{165956}$ directed edges.
We use two different settings for opinion updating rates:
one is to set $\lambda_i=1$ for all $i\in[n]$;
the other is to set $\lambda_i$ to be the number of $v_i$'s tweets in a year.
The users' inward probabilities are set in the following way so that it follows the
distribution we learned:
we sort all the inward probabilities learned in the last section among $12$ topics,
denoted as $\hat p_1 \leq \hat p_2 \leq \dots \leq \hat p_k$.
For each user $v_i$ in the Weibo graph, we select an integer $j$ from $\{1, 2,\ldots, k+1\}$ uniformly at random, 
and set $v_i$'s inward probability to a random real number in the following interval
$$\begin{cases}
[0, \hat p_1], &\text{if }j=1,\\
[\hat p_k, 1], &\text{if }j=k+1, \\
[\hat p_{j-1}, \hat p_j], &\text{others.}
\end{cases}$$
Since there are some $\hat p_j$ values that are zeros, we will have users with zero
inward probability.
For these users, we use a very small value $10^{-10}$ in our simulation since our computation of the \model{} model
requires inward probability to be greater than zero.
Figure \subref{fig:lambda=1} and \subref{fig:poisson} show the experimental results on the Weibo graph with all $\lambda_i=1$ and $\lambda_i$ set to the number of $v_i$'s tweets, respectively.
    \begingroup
    \allowdisplaybreaks
    \section{Mathematical Proofs}\label{section: proof}

\subsection{Unbiasedness of Partitioned Sampling}\label{unbiased proof}
\begin{proposition*}
    (Unbiasedness) Partitioned sampling is unbiased.
    Specifically, for any partition $\mathcal{P}$, $\E \left[\hat f_{\it part}(\mathcal{P})\right]=\bar f$.
\end{proposition*}
\begin{proof}
    For any partition $\mathcal{P}=\{(V_1, r_1),\ldots,(V_K,r_K)\}$,
    %according to the output formulation of Method \ref{al:partitioned},
    \begin{align*}
    &\E \left[\hat f_{\it part}(\mathcal{P})\right]
    = \sum_{k=1}^K  \frac{|V_k|}{|V|} \E \left[\hat f_{\it naive}(V_k,r_k)\right]
    =\frac{1}{|V|} \sum_{k=1}^K  |V_k| \sum_{v_i\in V_k}\frac{f(v_i)}{|V_k|}
    =\frac{\sum_{i=1}^n f(v_i)}{|V|}
    =\bar f.
    \end{align*}
    Notice that naive sampling in any group $V_k$ is unbiased, thus $\E [\hat f_{\it naive}(V_k,r_k)]$ is equal to the average opinion of the people in $V_k$ (second equality above).
    Therefore partitioned sampling is unbiased.
\end{proof}

\subsection{Proof of Theorem \ref{theorem:objective function}}
\objective*
\begin{proof}
    We use $x_k$ to denote the sample node selected in the $k$-th group $V_k$ of the simple partition $\mathcal{P}$. 
    The estimate of partitioned sampling with $\mathcal{P}$ can be written as
    $$\hat f_{\it part}(\mathcal P) = \frac{\sum_{k=1}^r n_k f(x_k)}{n},$$
    where $n=|V|$ and $n_k=|V_k|$.
	When $f$ is fixed, since $f(x_k)$'s for all $k \in [r]$	are independent, we have
	\[
	\Var_S \left(\hat f_{\it part}(\mathcal{P})\right) = \frac{1}{n^2} \sum_{k=1}^r n_k^2 \cdot \Var_S\left[f(x_k)  \right ]
		= \frac{1}{n^2} \sum_{k=1}^r n_k^2 \cdot (\E_S[f(x_k)^2] - \E_S[f(x_k)]^2).
	\]
	We then use the fact that $f(x_k)^2 = f(x_k)$, and $\E_S[f(x_k)] = \sum_{v_j \in V_k} f(v_j) / n_k$, to obtain
	\begin{align*}
	\Var_S \left(\hat f_{\it part}(\mathcal{P})\right) & = \frac{1}{n^2} \sum_{k=1}^r \left( n_k \sum_{v_j \in V_k} f(v_j) 
		- \left( \sum_{v_j \in V_k} f(v_j) + \sum_{v_i,v_j \in V_k, v_i\ne v_j} f(v_i)f(v_j)  \right)   \right)  \\
		& = \frac{1}{n^2} \sum_{k=1}^r \left( (n_k-1) \sum_{v_j \in V_k} f(v_j) 
		- \sum_{v_i,v_j \in V_k, v_i\ne v_j} f(v_i)f(v_j)  \right). 
	\end{align*}
	When $f$ is drawn from a distribution, we have
	\begin{align}\label{eq:evar}
	\E_M\left[\Var_S \left(\hat f_{\it part}(\mathcal{P})\right) \right] &= \frac{1}{n^2} \sum_{k=1}^r \left( (n_k-1) \sum_{v_j \in V_k} \E_M[f(v_j)] 
	- \sum_{v_i,v_j \in V_k, v_i\ne v_j} \E_M[f(v_i)f(v_j)]  \right).
	\end{align}
	Notice that for any two binary ($0/1$) random variables $A$ and $B$, the following equation holds:
	    \begin{align*}
	    \E[AB]=\frac{1}{2}\left(\Pr[A=B]+\E[A]+\E[B]-1\right).
	    \end{align*}
        Thus
            \begin{align}\label{eq: E[fifj]}
            \E_M[f(v_i) f(v_j)]=\frac{1}{2}(\s_{ij}+\E_M[f(v_i)]+\E_M[f(v_j)]-1).
            \end{align}
	Applying the above to Eq.~\eqref{eq:evar}, we have
	\begin{align*}
	&\E_M\left[\Var_S \left(\hat f_{\it part}(\mathcal{P})\right) \right]  \\
	&  = \frac{1}{n^2} \sum_{k=1}^r \left( (n_k-1) \sum_{v_j \in V_k} \E_M[f(v_j)] 
	+ \frac{1}{2}\cdot \sum_{v_i,v_j \in V_k, v_i\ne v_j} (1-\sigma_{ij}) - \frac{1}{2}\cdot \sum_{v_i,v_j \in V_k, v_i\ne v_j} (\E_M[f(v_i)] + \E_M[f(v_j)] ) \right) \\
	& = \frac{1}{2n^2} \sum_{k=1}^r \sum_{v_i,v_j \in V_k, v_i\ne v_j} (1-\sigma_{ij}) = \frac{1}{2n^2} \sum_{k=1}^r \Vol_{G_a}(V_k) = \frac{1}{2|V|^2} \cdot g({\mathcal{P}}).
	\end{align*}
\end{proof}

\subsection{Proof of Theorem \ref{theorem: simple is better}}

\greedy*
\begin{proof}
It has been known that, the sample variance of naive sampling is
\begin{equation}\label{eq:naive variance}
\Var_S(\hat f_{\it naive}(V,r))=\frac{1}{r}(\bar f-\bar f^2)
\end{equation}
where $\bar f$ is the average opinion of the entire population.
Thus
$$\E_M[\Var_S(\hat f_{\it naive})]
= \frac{\E_M[\bar f]-\E_M[\bar f^2]}{r},$$
where $\E_M[\bar f]=\frac{1}{n}\sum_{i=1}^n \E_M[f(v_i)]$, and
\begin{align*}
\E_{M}\left[\bar f^2\right]
&= \E_{M}\left[\left(\frac{\sum_{i=1}^n f(v_i)}{n}\right)^2\right] \\
&=\frac{1}{n^2}\sum_{i=1}^n \E_{M} [f(v_i)] +\frac{2}{n^2}\sum_{i<j}\E_M[f(v_i) f(v_j)] \\
&=\frac{1}{n}\sum_{i=1}^n \E_{M} [f(v_i)]-\frac{1}{n^2}\sum_{i<j}(1-\s_{ij}). \text{  (using Eq. (\ref{eq: E[fifj]}))}
\end{align*}
Therefore %the expected variance of naive sampling is
\begin{align*}
\E_M[\Var_S(\hat f_{\it naive})]
= \frac{\E_M[\bar f]-\E_M[\bar f^2]}{r}
= \frac{\sum_{i \neq j}(1-\s_{ij})}{2n^2r}.
%= \frac{\sum_{i \neq j} w_{ij}}{2n^2r}
\end{align*}

For the greedy partitioning algorithm (Algorithm \ref{al:greedy}), 
    suppose the randomly generated node sequence is
    $v_{s_1}$, $v_{s_2}$, \dots, $v_{s_n}$ (which is denoted by $x_1,x_2,\dots,x_n$ in the algorithm).
For the greedy assignment of $k$-th node in the first iteration, the group that $v_{s_k}$ is assigned to should make the cost function increased the least, 
    thus the increase of cost function is no more than $2\sum_{l=1}^{k-1} w_{s_k s_l}/r$, 
    where $w_{s_k s_l}=1-\sigma_{s_k s_l}$ is the weight of edge $(v_{s_k},v_{s_l})$ in the graph $G_a$.
Thus at the end of the first iteration, the cost function
\begin{align*}
g(\mathcal{P})
\leq 2 \sum_{k=2}^r\sum_{l=1}^{k-1} w_{s_k s_l}/r
= \frac{1}{r}\sum_{i \neq j}(1-\s_{ij}).
\end{align*}
Therefore
\begin{align*}
\E_M[\Var_S(\hat f_{\it part}(\mathcal{P}))]-\E_M[\Var_S(\hat f_{\it naive}(V,r))] 
=\frac{1}{2n^2}g(\mathcal{P})-\frac{\sum_{i \neq j}(1-\s_{ij})}{2n^2r}
\leq 0.
\end{align*}
This finishes the proof.
\end{proof}

\simple*
\begin{proof}
    We first show that the sample variance of partitioned sampling can be written as a weighted summation of the sample variance of naive sampling in each group as below:
\begin{align}\label{eq: part variance}
\Var_S(\hat f_{\it part}(\mathcal{P}))
=\sum_{k=1}^K \frac{n_k^2}{n^2} \Var_S(\hat f_{\it naive}(V_k,r_k))
\end{align}
where $n=|V|$, $n_k=|V_k|$, and $K$ is the number of groups of $\mathcal{P}$.

According to the definition of partitioned sampling,
$$\hat f_{\it part}(\mathcal{P})=\sum_{k=1}^K \frac{n_k}{n} \hat f_{\it naive}(V_k,r_k).$$
The estimate of naive sampling in two different groups are independent,
thus for any $k \neq l$,
\begin{align*}
\E_S[\hat f_{\it naive}(V_k,r_k) \hat f_{\it naive}(V_l,r_l)] 
=\E_S[\hat f_{\it naive}(V_k,r_k)]\cdot \E_S[\hat f_{\it naive}(V_l,r_l)].
\end{align*}
Therefore
\begin{align*}
\Var_S(\hat f_{\it part}(\mathcal{P}))
&=\E_S[\hat f_{\it part}(\mathcal{P})^2]-\E_S[\hat f_{\it part}(\mathcal{P})]^2\\
&=\E_S\hspace{-1mm}\left[\hspace{-1mm}\left(\sum_{k=1}^K \frac{n_k}{n} \hat f_{\it naive}(V_k,r_k)\hspace{-1mm}\right)^2\right]
\hspace{-1mm}-\hspace{-0.5mm}\E_S\hspace{-1mm}\left[\sum_{k=1}^K \frac{n_k}{n} \hat f_{\it naive}(V_k,r_k)\right]^2\\
&=\sum_{k=1}^K \frac{n_k^2}{n^2}\E_S[\hat f_{\it naive}(V_k,r_k)^2]
+ \sum_{k\neq l}\frac{n_k n_l}{n^2}\E_S[\hat f_{\it naive}(V_k,r_k) \hat f_{\it naive}(V_l,r_l)]\\
&\quad -\sum_{k=1}^K \frac{n_k^2}{n^2}\E_S[\hat f_{\it naive}(V_k,r_k)]^2
- \sum_{k\neq l}\frac{n_k n_l}{n^2}\E_S[\hat f_{\it naive}(V_k,r_k)]\E_S[ \hat f_{\it naive}(V_l,r_l)]\\
&=\sum_{k=1}^K \frac{n_k^2}{n^2}\E_S[\hat f_{\it naive}(V_k,r_k)^2]
-\sum_{k=1}^K \frac{n_k^2}{n^2}\E_S[\hat f_{\it naive}(V_k,r_k)]^2\\
&=\sum_{k=1}^K \frac{n_k^2}{n^2} \Var_S(\hat f_{\it naive}(V_k,r_k)).
\end{align*}
   
    For any partition $\mathcal{P}$ with $K$ groups, if there exists some group $V_k$ containing more than one sample nodes, according to Lemma \ref{lemma:greedy vs naive}, we can efficiently find the simple partition $\mathcal{P}_k^*$ for that group $V_k$ by one-round greedy partitioning algorithm such that
    \[
    \E_M\left[\Var_S(\hat f_{\it part}({\cal P}_k^*))\right] \leq
    \E_M\left[\Var_S(\hat f_{\it naive}(V_k,r_k))\right].
    \]
    
    Thus we do the above refining procedure for all the groups containing more than one sample nodes iteratively,
    and combine all the final groups together to get 
    the refined simple partition $\mathcal{P}'$ of $\mathcal{P}$. 
    It satisfies that
    \begin{align*}
    \E_M\left[\Var_S(\hat f_{\it part}({\cal P}'))\right] 
    & =\sum_{k=1}^K \frac{n_k^2}{n^2}\E_M\left[\Var_S(\hat f_{\it part}({\cal P}_k^*))\right] \\
    &\leq \sum_{k=1}^K \frac{n_k^2}{n^2} \E_M\left[\Var_S(\hat f_{\it naive}(V_k,r_k))\right] \\
    &=\E_M\left[\Var_S(\hat f_{\it part}({\cal P}))\right].
    \end{align*}
    
    This means that partitioned sampling using the refined simple partition $\cal P'$ is 
    at least as good as partitioned sampling using the original partition $\cal P$.
\end{proof}

\subsection{Proof of Theorem \ref{thm:balancedpartition}}
\balanced*
\begin{proof}    
    According to Eq. (\ref{eq: part variance}), for any simple partition $\mathcal{P}$, we have
    \begin{align*}
    \Var_S(\hat f_{\it part}(\mathcal{P}))=\sum_{k=1}^r \frac{n_k^2}{n^2}\Var_S(\hat f_{\it naive}(V_k, 1)).
    \end{align*}
    When $\mathcal{P}$ is a balanced simple partition, we have $n_k=n/r$ and $r_k=1$, thus
    \begin{align*}
    \Var_S(\hat f_{\it part}(\mathcal{P}))=\frac{1}{r^2}\sum_{k=1}^r \Var_S(\hat f_{\it naive}(V_k, 1))=\frac{1}{r^2}\sum_{k=1}^r \bar f_k(1-\bar f_k) \text{ (using Eq. (\ref{eq:naive variance}))}
    \end{align*}
    where $\bar f_k$ is the average opinion of the $k$-th group.
    
    Notice that $\bar f=\sum_{k=1}^r \bar f_k/r$ holds for any balanced simple partition, thus we have
    \begin{align*}
    \Var_S\left(\hat f_{\it part}({\cal P})\right)-\Var_S\left(\hat f_{\it naive}(V,r)\right)  
    &=\frac{1}{r^2}\sum_{k=1}^r \bar f_k(1-\bar f_k)-\frac{1}{r}\bar f (1-\bar f) \\
    &=\frac{1}{r^2}\sum_{k=1}^r \bar f_k(1-\bar f_k)-\frac{1}{r^2}\left(\sum_{k=1}^r\bar f_k\right) \left(1-\frac{1}{r}\sum_{k=1}^r \bar f_k\right)\\
    &=-\frac{1}{r^3}\left[ r{\sum_{k=1}^r\bar f_k^2}-{\left(\sum_{k=1}^r \bar f_k\right)^2} \right].
    \end{align*}
    
    According to Cauchy-Schwartz inequality, 
    $$r{\sum_{k=1}^r\bar f_k^2}=\left(\sum_{k=1}^r 1^2\right)\cdot \left(\sum_{k=1}^r\bar f_k^2\right)\geq{\left(\sum_{k=1}^r 1\cdot \bar f_k\right)^2}=\left(\sum_{k=1}^r \bar f_k\right)^2.$$
    Thus we have
    $$\Var_S\left(\hat f_{\it part}({\cal P})\right)-\Var_S\left(\hat f_{\it naive}(V,r)\right) \leq 0.$$
    This finishes the proof.
\end{proof}

\subsection{Proof of Theorem \ref{theorem: similarity}}
\walkparam*
\begin{proof}
    (a) For parameter ${\cal I}_{ij}^\ell$.    
    
    Recall from Definition \ref{def:walkparam} that ${\cal I}_{ij}^\ell$ denotes the event that two random walks starting from $v_i$ and $v_j$ at time $t=\infty$ eventually meet and the first node they meet at is $v_\ell \in V$.
    This event consists of two steps:
    1) the walker at $v_i$ (or $v_j$) moves to one of its neighbor $v_a$ (or $v_b$) with probability $(1-p_i)A_{ia}/d_i$ (or $(1-p_j)A_{jb}/d_j$);
    2) two random walks starting from $v_a$ (or $v_b$) and $v_j$ (or $v_i$) eventually meet and the first node they meet at is $v_\ell$.
    The probability that the walker at $v_i$ (resp. $v_j$) make a movement first is proportional to $v_i$'s (resp. $v_j$'s) Poisson rate, that is $\lambda_i/(\lambda_i+\lambda_j)$ (resp. $\lambda_j/(\lambda_i+\lambda_j)$).
    Thus
    for any $i\neq j$, $\Pr\left[{\cal I}_{ij}^\ell\right]$ can be calculated by the following recursion:
    \begin{align*}
        \Pr\left[{\cal I}_{ij}^\ell\right]
        &= \sum_{a=1}^n \frac{\lambda_i}{\lambda_i+\lambda_j}\frac{(1-p_i)A_{ia}}{d_i}\Pr\left[{\cal I}_{aj}^\ell\right] 
         + \sum_{b=1}^n \frac{\lambda_j}{\lambda_i+\lambda_j}\frac{(1-p_j)A_{jb}}{d_j}\Pr\left[{\cal I}_{ib}^\ell\right].
    \end{align*}
    %\begin{eqnarray*}
    %    \Pr\left[{\cal I}_{ij}^\ell\right]
    %    &=&\sum_{v_a\in N_i, v_b\in N_j} \Pr\left[v_i\to v_a, v_j \to v_b\right] \Pr\left[{\cal I}_{ab}^\ell\right] \\
    %    &=&\sum_{v_a\in N_i, v_b\in N_j} \frac{(1-p_i)(1-p_j) A_{ia} A_{jb}} {d_i d_j} \Pr\left[{\cal I}_{a b}^\ell\right].
    %    \end{eqnarray*}
    When $i=j$, $\Pr\left[{\cal I}_{ij}^\ell\right]$ is determined by the following boundary conditions:
    $$\Pr\left[{\cal I}_{ij}^\ell\right]=
    \begin{cases}
    0, & i=j\neq \ell,\\
    1, & i=j=\ell.
    \end{cases}$$
    By combining the recursive equations and the boundary conditions, we have the following linear equations:
    
    \begin{equation}\label{eq:I_ij^l}
        \Pr\left[{\cal I}_{ij}^\ell\right]=
        \begin{cases}
            0, & i=j\neq \ell, \\
            1, & i=j=\ell,\\
            \sum_{a=1}^n \frac{\lambda_i (1-p_i) A_{ia} } {(\lambda_i+\lambda_j)d_i} \Pr[{\cal I}_{a j}^\ell]
            +\sum_{b=1}^n \frac{\lambda_j (1-p_j) A_{jb} } {(\lambda_i+\lambda_j)d_j} \Pr[{\cal I}_{i b}^\ell]
            , & i\neq j.
        \end{cases}
    \end{equation}
    The above proof follows the idea in \cite{DiscreteOpinion}.
       
    Next we show that the linear system (\ref{eq:I_ij^l}) has a unique solution.  
    For a fixed $\ell$, the equations for all terms $\Pr\left[{\cal I}_{ij}^\ell\right]$ such that $i\neq j$ form a linear sub-system with $\binom{n}{2}$ variables and $\binom{n}{2}$ equations, since $\Pr\left[{\cal I}_{ij}^\ell\right]=\Pr\left[{\cal I}_{ji}^\ell\right]$.
    Therefore, we can solve the whole linear system (\ref{eq:I_ij^l}) by solving $n$ separated linear sub-systems.
    Each linear sub-system corresponds to a value of $\ell$, and it can be solved\footnote{$n$-variable linear system can be solved in time $O\left(n^{3}\right)$.} in $O\left(\binom{n}{2}^{3}\right)=O\left(n^{6}\right)$, thus the original linear system (\ref{eq:I_ij^l}) can be solved in time $n\cdot O\left(n^6\right)=O\left(n^7\right)$. Thus we further develop an efficient computation method as shown in Section \ref{section: correlation computing proof}. 
    
    Now we show that there exists the unique solution for each linear sub-system.    
    Each equation in the linear sub-system can be written as    
    \begin{align*}
        \Pr\left[{\cal I}_{ij}^\ell\right]=\sum_{a\neq j} &\frac{\lambda_i(1-p_i)A_{ia}}{(\lambda_i+\lambda_j)d_i}\Pr\left[{\cal I}_{aj}^\ell\right]
        +\sum_{b\neq i} \frac{\lambda_j(1-p_j)A_{jb}}{(\lambda_i+\lambda_j)d_j}\Pr\left[{\cal I}_{ib}^\ell\right] 
        + \frac{\lambda_i(1-p_i)A_{ij}}{(\lambda_i+\lambda_j)d_i} \cdot 1_{j=\ell}
        +\frac{\lambda_j(1-p_j)A_{ji}}{(\lambda_i+\lambda_j)d_j} \cdot 1_{i=\ell}.
    \end{align*}
    
    Let $k=h(i,j)=(i-1)n+j$, then we have a bijection $h$ of subscript between integer $k$ and ordered pair $(i,j)$ where $i<j$. Then we can write these $\binom{n}{2}$ equations in the matrix form:
    $$I\vec x = M_1\vec x + M_2 \vec x+\vec b$$
    where $\vec x$ is a $\binom{n}{2}\times 1$ vector whose $k$-th element is $\Pr\left[{\cal I}_{ij}^\ell\right]$; 
    $M_1$ is a $\binom{n}{2}\times \binom{n}{2}$ matrix whose $(k,h(a,j))$ entry is $\frac{\lambda_i(1-p_i)A_{ia}}{(\lambda_i+\lambda_j)d_i}$; 
    $M_2$ is a $\binom{n}{2}\times \binom{n}{2}$ matrix whose $(k,h(i,b))$ entry is $\frac{\lambda_j(1-p_j)A_{jb}}{(\lambda_i+\lambda_j)d_j}$; 
    $\vec b$ is a $\binom{n}{2}\times 1$ vector whose $k$-th element is $\frac{\lambda_i(1-p_i)A_{ij}}{(\lambda_i+\lambda_j)d_i} \cdot 1_{j=\ell}
    +\frac{\lambda_j(1-p_j)A_{ji}}{(\lambda_i+\lambda_j)d_j} \cdot 1_{i=\ell}$.
    
    If $I-M_1-M_2$ is non-singular, then each linear sub-system has a unique solution $(I-M_1-M_2)^{-1}\vec b$. 
    In fact, for any row $k$ of $I-M_1-M_2$, let $(i,j)=h^{-1}(k)$, and the sum of the absolute value of the $k$-th row except the diagonal entry
    \begin{align*}
        \sum_{t:t\neq k} |I-M_1-M_2|_{kt}
        &= \sum_{a:a\neq j} \frac{\lambda_i(1-p_i)A_{ia}}{(\lambda_i+\lambda_j)d_i}
        +\sum_{b:b\neq i} \frac{\lambda_j(1-p_j)A_{jb}}{(\lambda_i+\lambda_j)d_j}\\
        &= \frac{\lambda_i(1-p_i)}{(\lambda_i+\lambda_j)}\sum_{a:a\neq j}\frac{A_{ia}}{d_i}
        + \frac{\lambda_j(1-p_j)}{(\lambda_i+\lambda_j)}\sum_{b:b\neq i}\frac{A_{jb}}{d_j}\\
        &\leq \frac{\lambda_i(1-p_i)}{(\lambda_i+\lambda_j)}
        + \frac{\lambda_j(1-p_j)}{(\lambda_i+\lambda_j)}\\
        & < \frac{\lambda_i}{(\lambda_i+\lambda_j)}
        + \frac{\lambda_j}{(\lambda_i+\lambda_j)}\\
        &
        =|I-M_1-M_2|_{kk}.
    \end{align*}
    Thus, $I-M_1-M_2$ is strictly diagonally dominant. According to the Levy-Desplanques theorem, it is non-singular.
    Since each linear sub-system has one unique solution, the whole linear system~(\ref{eq:I_ij^l}) also does.

    (b) For parameter $Q$.
    
    The probability of a walker from $v_i$ walking to $v_j$ in one walk step is $p_{ij}^{(1)}=(1-p_i)A_{ij}/d_i$.
    So we have a matrix form $P_{VV}=(I-P)D^{-1}A$ whose $(i,j)$ entry is $p_{ij}^{(1)}$. Therefore, the probability of walking from $v_i$ to $v_j$ in exactly $\ell$ steps is the $(i,j)$ entry of $(P_{VV})^\ell$. It is easy to verify that $\lim_{\ell\rightarrow+\infty}(P_{VV})^\ell \rightarrow 0$.
    
    By definition of our model, the probability of $v_j$ walking to $v_j'$ (being absorbed) is $p_j$.
    Thus the matrix $Q$ whose $(i,j)$ entry is the probability of transition from $v_i$ to $v_j'$ can be calculated by
    \begin{align*}
        Q&=\sum_{\ell=0}^{\infty} (P_{VV})^\ell P=(I-P_{VV})^{-1}P
        =\left(I-\left(I-P\right)D^{-1}A\right)^{-1}P.
    \end{align*}
    Now we show that $I-P_{VV}$ is invertible. 
    The $(i,j)$ entry of $I-P_{VV}$ is
    $$\begin{cases}1, & \text{if }i=j, \\ -\frac{1-p_i}{d_i}A_{ij}, & \text{if } i\neq j. \end{cases}$$
    For any row $i$ of $I-P_{VV}$, the sum of absolute values of its non-diagonal elements can be written as $\sum_{j:j\neq i} \frac{1-p_i}{d_i}A_{ij}=(1-p_i)(1-\frac{A_{ii}}{d_i})$, and it is strictly less than the absolute value of $i$-th diagonal elements $\left|I-P_{VV}\right|_{ii}=1$. 
    Thus $I-P_{VV}$ is strictly diagonally dominant, and it is non-singular (Levy-Desplanques theorem).
\end{proof}

\correlation*
\begin{proof}
    Before proving the lemma, we first introduce the following proposition:
    \begin{proposition*}
        The expected expressed opinion of each node in the steady state
        is equal to the expected value of innate opinion, namely, for all
        $i\in [n]$,
        \begin{equation}\label{eq:expectation}
        \mu_i=\E_M[f^{(\infty)}(v_i)]=\mu^{(0)}.
        \end{equation}
    \end{proposition*}

        We prove this by proving a stronger statement: given any $t\ge 0$, $\E_M [f^{(t)}(v_i)]=\mu^{(0)}$ for all $i\in[n]$. Namely, we want to prove that at any time $t$, each node's expected expressed opinion is equal to the expected innate opinion.  
        The proof is by induction on time $t$.
        In the initial state, each node's expressed opinion (also innate opinion) is generated from an i.i.d.\ distribution with expected value $\mu^{(0)}$, and thus the above statement holds.
        Now suppose the statement holds before time $t$. It still holds before the next Poisson arrival among all nodes.
        Suppose the next Poisson arrival comes at time $t_1$ and its corresponding updating node is $v_i$.
        At this Poisson arrival time $t_1>t$, $v_i$ updates its expressed opinion based on its innate opinion and one of its neighbors' expressed opinions. 
        Notice that the expectations of both its innate opinion $f^{(0)}(v_i)$ and all its neighbors' expressed opinions $f^{(t_1)}(v_j)$ are equal to $\mu^{(0)}$ by the inductive assumption, namely,
        $\E_M[f^{(0)}(v_i)]=\mu$ and $\E_M[f^{(t_1)}(v_j)]=\mu$ for all $v_j\in N_i$, where $N_i$ is the set of $v_i$'s neighbors.
        Thus the expectation of $v_i$'s updated expressed opinion $\E_M[f^{(t_1)}(v_i)]$ is still equal to $\mu$.
        Moreover, other nodes' expected expressed opinions remain  equal to $\mu^{(0)}$ upon time $t_1$.
        Thus by induction, at any time $t$, each node's expected expressed opinion is always equal to $\mu^{(0)}$.

    Now we are ready to prove Lemma \ref{lemma:correlation}.
    
    (a) In this part, we show that the opinion correlation $\Cor_M(f^{(\infty)}(v_i), f^{(\infty)}(v_j))$ is 
    equal to the probability that two coalescing random walks starting from $v_i$ and $v_j$ at time $t=\infty$ end at the same absorbing node in $V^\prime$.
    In the proof, we split the randomness $M$ into two parts: we use $O$ to denote the randomness of innate opinions which are generated from an i.i.d.\ distribution, and we use $E$ to denote the randomness from the opinion evolution. 
    
    When $i=j$, obviously we have $\c_{ij}=1$. 
    In this case, the two random walks' paths coincide, thus they are absorbed by the same node in $V^\prime$ with probability $1$. 
    
    When $i\neq j$, according to the definition of correlation,
    \begin{align*}
        \c_{ij}&=\Cor_M(f^{(\infty)}(v_i), f^{(\infty)}(v_j))\\
        &=\frac{\E_{M}[f^{(\infty)}(v_i)f^{(\infty)}(v_j)]-\E_{M}[f^{(\infty)}(v_i)]\E_{M}[f^{(\infty)}(v_j)]}
        {\sqrt{\Var_{M}[f^{(\infty)}(v_i)]\Var_{M}[f^{(\infty)}(v_j)]}}\\
        &=\frac{\E_{M}[f^{(\infty)}(v_i)f^{(\infty)}(v_j)]-(\mu^{(0)})^2}{\mu^{(0)}-(\mu^{(0)})^2}.\numberthis\label{eq: Cor and E}
    \end{align*}
    The third equality holds because for any $i\in[n]$,
    \begin{align*}
    \Var_{M}[f^{(\infty)}(v_i)]
    =\E_{M}[f^{(\infty)}(v_i)^2]-\E_{M}[f^{(\infty)}(v_i)]^2
    =\E_{M}[f^{(\infty)}(v_i)]-\E_{M}[f^{(\infty)}(v_i)]^2=\mu^{(0)}-(\mu^{(0)})^2.
    \end{align*}
    Next, we need to calculate $\E_{M}[f^{(\infty)}(v_i)f^{(\infty)}(v_j)]$, which is the probability that two random walkers starting from $v_i$ and $v_j$ walk to the nodes in $V^\prime$ whose original innate opinions are $1$. 
    This event consists of two cases: two random walkers move to the same absorbing node, or two distinct absorbing nodes. 
    Thus we can calculate $\E_{M}[f^{(\infty)}(v_i)f^{(\infty)}(v_j)]$ by adding them together.
    
    Let ${\cal M}_{i,j}^{p,q}$ be the event that in the coalescing random walks on $\overline G$, a random walker starting from $v_i$ is absorbed by $v_p^\prime$, while another random walker starting from $v_j$ is absorbed by $v_q^\prime$. Note that ${\cal M}_{i,j}^{p,q}$ is measurable under randomness $E$. It only depends on the structure of $\overline G$ and is independent of the initial value in $V'$:
    \begin{equation*}
        \Pr_E\left[{\cal M}_{i,j}^{p,q}\middle| f^{(0)}(v_1), f^{(0)}(v_2), \cdots, f^{(0)}(v_n)\right] =  \Pr_E\left[{\cal M}_{i,j}^{p,q}\right].
    \end{equation*}
    Thus $\E_{M}[f^{(\infty)}(v_i)f^{(\infty)}(v_j)]$ can be written as:
    \begin{align*}
        \E_{M}[f^{(\infty)}(v_i)f^{(\infty)}(v_j)]
        &=\E_{O,E}[f^{(\infty)}(v_i)f^{(\infty)}(v_j)]\\
        &=\sum_{p \neq q}\left(\Pr_{O,E}\left[{\cal M}_{i,j}^{p,q}\mid f^{(0)}(v_p) f^{(0)}(v_q) = 1\right]
         \Pr_{O,E}\left[f^{(0)}(v_p) f^{(0)}(v_q) = 1\right]\right)\\
        &\quad + \sum_{p=1}^n\Pr_{O,E}\left[{\cal M}_{i,j}^{p,p}\mid f^{(0)}(v_p) = 1\right]\Pr_{O,E}\left[f^{(0)}(v_p) = 1\right] \\
        &=\sum_{p \neq q}\Pr_{E}\left[{\cal M}_{i,j}^{p,q}\right]\Pr_{O}\left[f^{(0)}(v_p) f^{(0)}(v_q) = 1\right] 
        + \sum_{p=1}^n\Pr_{E}\left[{\cal M}_{i,j}^{p,p}\right]\Pr_{O}\left[f^{(0)}(v_p) = 1\right]\\
        &= (\mu^{(0)})^2 \sum_{p \neq q}\Pr_{E}\left[{\cal M}_{i,j}^{p,q}\right]+  \mu^{(0)}\sum_{p=1}^n\Pr_{E}\left[{\cal M}_{i,j}^{p,p}\right]\\
        &=\mu^{(0)}\sum_{p=1}^n\Pr_{E}\left[{\cal M}_{i,j}^{p,p}\right] + (\mu^{(0)})^2 \left(1-\sum_{p=1}^n\Pr_{E}\left[{\cal M}_{i,j}^{p,p}\right]\right) \\
        &=:\mu^{(0)} p_{same}(i,j)+(\mu^{(0)})^2(1-p_{same}(i,j)).
    \end{align*}
    In the last equation, we use $p_{same}(i,j)$ to denote the probability that two coalescing random walks starting from $v_i$ and $v_j$ end at the same node in $V^\prime$. 
    Thus
    \begin{align*}
        \c_{ij}%&=\Cor_M(f^{(\infty)}(v_i), f^{(\infty)}(v_j))\\
        &=\frac{\E_{M}[f^{(\infty)}(v_i)f^{(\infty)}(v_j)]-(\mu^{(0)})^2}{\mu^{(0)}-(\mu^{(0)})^2}\\
        &=\frac{\left[\mu^{(0)} p_{same}(i,j)+(\mu^{(0)})^2(1-p_{same}(i,j))\right]-(\mu^{(0)})^2}{\mu^{(0)}-(\mu^{(0)})^2}\\
        &=p_{same}(i,j).
    \end{align*}
    This means that the opinion correlation $\c_{ij}$ is 
    equal to the probability that two coalescing random walks starting from $v_i$ and $v_j$ end at the same absorbing node in $V^\prime$.
    
    (b) We now calculate $p_{same}$ in this part. 
    Let ${\cal H}_{ij}^k$ be the event that two coalescing random walks starting from $v_i$ and $v_j$ are both absorbed by node  $v_k^\prime$ without meeting each other at a node in $V$. 
    According to the definitions of events ${\cal M}_{i,j}^{p,q}$, ${\cal I}_{ij}^\ell$ and ${\cal H}_{ij}^k$, we have
    \begin{equation}\label{eq:1}
        \Pr_E\left[{\cal M}_{i,j}^{k,k}\right]=\sum_{\ell=1}^n\Pr\left[{\cal I}_{ij}^\ell\right]Q_{\ell k}+\Pr\left[{\cal H}_{ij}^k\right]
    \end{equation}
    where $Q_{\ell k}$ 
    is the probability
    that a random walker starting from node $v_\ell$ at time 
    ends at $v_k^\prime \in V'$.
    
    %is the $(\ell,k)$ entry of matrix $Q=(I-(I-P)D^{-1}A)^{-1}P$. 
    
    Notice that $Q_{ik}Q_{jk}$ represents the probability that two {\em non}-coalescing random walks starting from $v_i$ and $v_j$ end at node $v_k^\prime$, thus it can be written as
    \begin{equation}\label{eq:2}
        Q_{ik}Q_{jk}=\sum_{\ell=1}^n\Pr\left[{\cal I}_{ij}^\ell\right]Q_{\ell k}^2+\Pr\left[{\cal H}_{ij}^k\right].
    \end{equation}
    Combining Eq. (\ref{eq:1}) and (\ref{eq:2}), 
    \begin{equation}\label{eq:3}
        \Pr_E\left[{\cal M}_{i,j}^{k,k}\right] = \sum_{\ell=1}^n\Pr\left[{\cal I}_{ij}^\ell\right]\left(Q_{\ell k}-Q_{\ell k}^2\right)+Q_{ik}Q_{jk}.
    \end{equation}
    Therefore,
    \begin{align*}
        \c_{ij}
        &=p_{\it same}(i,j)=\sum_{k=1}^n \Pr_E\left[{\cal M}_{i,j}^{k,k}\right] \\
        &=\sum_{k=1}^n \left(\sum_{\ell=1}^n\Pr\left[{\cal I}_{ij}^\ell\right]\left(Q_{\ell k}-Q_{\ell k}^2\right)+Q_{ik}Q_{jk}\right)\\
        &=\sum_{\ell=1}^n \Pr\left[{\cal I}_{ij}^\ell\right] \left(1-\sum_{k=1}^n Q_{\ell k}^2\right)+\sum_{k=1}^n Q_{ik}Q_{jk}.
    \end{align*}
    This finishes the proof.
\end{proof}

\similarity*
\begin{proof}
    Recall Eq. (\ref{eq: E[fifj]}) and Eq. (\ref{eq:expectation}),
    \begin{align*}
        \E_M[f^{(\infty)}(v_i) f^{(\infty)}(v_j)]
        &=\frac{1}{2}(\s_{ij}+\E_M[f^{(\infty)}(v_i)]+\E_M[f^{(\infty)}(v_j)]-1)\\
        &=\frac{1}{2}(\s_{ij}-1)+\mu^{(0)},
    \end{align*}
    and Eq.~(\ref{eq: Cor and E}),
    \begin{align*}
        \E_M[f^{(\infty)}(v_i) f^{(\infty)}(v_j)]
        =\mu^{(0)}(1-\mu^{(0)})\c_{ij}+(\mu^{(0)})^2.
    \end{align*}
    Thus we can obtain
    \begin{align*}
        \s_{ij}=1-2\mu^{(0)}(1-\mu^{(0)})(1-\c_{ij}).
    \end{align*}
    This finishes the proof.
\end{proof}
    \section{Extensions of the VIO Model}\label{section: more on vio}
In this section, we extend the VIO model to allow (a) non-i.i.d. distributions of the innate opinions, and (b) negative edges as in
the signed voter model \cite{Li0WZ15}. We provide the analysis of opinion similarities of these two extended models.
	\subsection{\model{} Model with Non-Identical Innate Opinion Distribution}\label{section: non iid}
	
	In the \model{} model (section \ref{section: opinion evolution model and opinion similarity}), we regard the individuals' innate opinions as generated from an i.i.d.\ distribution.
    In this section,
	we still assume that the individuals' innate opinions are independent, but each person $v_i$ can have her own expected innate opinion $\E[f^{(0)}(v_i)]=\mu_i^{(0)}$.
    In order to distinguish the source of the randomness, we split the randomness $M$ into two parts: we use $O$ to denote the randomness of innate opinions which are generated by the Bernoulli distribution, and use $E$ to denote the randomness from the opinion evolution (These notations were first introduced in the proof of Lemma \ref{lemma:correlation}).
	Thus $\E[f^{(0)}(v_i)]=\mu_i^{(0)}$ is clarified as $\E_O[f^{(0)}(v_i)]=\mu_i^{(0)}$.
    
	We now give a sketchy analysis of the \model{} model with non-identical innate opinion distribution.
	We still use notation $\e_i$ to denote $v_i$'s expected expressed opinion when the evolution converges, that is $\e_i=\E_{M}[f^{(\infty)}(v_i)]$.
    Notice that the computation method of $Q$ and $\{\Pr[\mathcal{I}_{ij}^\ell]\}$'s (Lemma \ref{lemma:walkparam}) still holds for the \model{} model with non-identical innate opinion distribution.
	
    According to the definition of matrix $Q$ (Definition \ref{def:walkparam}), it is easy to verify that for any $i\in[n]$,
    \begin{align*}
        \E_E[f^{(\infty)}(v_i)]=\sum_{j=1}^n Q_{ij} f^{(0)}(v_j).
    \end{align*}
    Thus
    \begin{align}\label{eq: non-iid expectation}
        \e_i=\E_O\E_E[f^{(\infty)}(v_i)]
        =\sum_{j=1}^n Q_{ij} \E_O [f^{(0)}(v_j)]=\sum_{j=1}^n Q_{ij} \mu_j^{(0)}.
    \end{align}
    
    Next, we focus on computing the opinion similarity between any pair of nodes $(v_i,v_j)$, 
        which is the probability that two random walkers starting from $v_i$ and $v_j$ have the same final expressed opinions.
    Recall Eq.~(\ref{eq: E[fifj]}) that
    \begin{align*}
        \s_{ij}&=2\E_M[f^{(\infty)}(v_i)f^{(\infty)}(v_j)]+1 -\e_i-\e_j.
    \end{align*}
    In order to obtain $\s_{ij}$, we compute $\E_M[f^{(\infty)}(v_i)f^{(\infty)}(v_j)]$, 
    which is the probability that two random walkers starting from $v_i$ and $v_j$ finally walk to the nodes in $V'$ whose innate opinions are equal to one.
    
    Following the notations used in the proof of Lemma \ref{lemma:correlation}, let ${\cal M}_{i,j}^{p,q}$ be the event that in the coalescing random walks on $\overline G$, a random walker starting from $v_i$ is absorbed by $v_p^\prime$, while another random walker starting from $v_j$ is absorbed by $v_q^\prime$. Note that ${\cal M}_{i,j}^{p,q}$ is measurable under randomness $E$. 
    Thus
    \begin{align*}
        \E_M[f^{(\infty)}(v_i)f^{(\infty)}(v_j)]
        &=\Pr_M[f^{(\infty)}(v_i)=f^{(\infty)}(v_j)=1]\\
        &=\sum_{p=1}^n \Pr_E[\mathcal{M}_{i,j}^{p,p}]\Pr_O[f^{(0)}(v_p)=1]
       +\sum_{p\neq q} \Pr_E[\mathcal{M}_{i,j}^{p,q}]\Pr_O[f^{(0)}(v_p)=1,f^{(0)}(v_q)=1]\\
        &=\sum_{p=1}^n \Pr_E[\mathcal{M}_{i,j}^{p,p}]\mu_p^{(0)}
         +\sum_{p\neq q} \Pr_E[\mathcal{M}_{i,j}^{p,q}]\mu_p^{(0)}\mu_q^{(0)}.
    \end{align*}
    Recall Eq (\ref{eq:3}),
    \begin{align*}
        \Pr_E\left[{\cal M}_{i,j}^{p,p}\right] = \sum_{\ell=1}^n\Pr\left[{\cal I}_{ij}^\ell\right]\left(Q_{\ell p}-Q_{\ell p}^2\right)+Q_{ip}Q_{jp},
    \end{align*}
     and use the same technique of Eq (\ref{eq:2}),
     \begin{align*}
         \Pr_E\left[{\cal M}_{i,j}^{p,q}\right]= Q_{ip}Q_{jq}-\sum_{\ell=1}^n\Pr\left[{\cal I}_{ij}^\ell\right]Q_{\ell p}Q_{\ell q}.
     \end{align*}
     Therefore, after some calculations,
     \begin{align*}
         \E_M[f^{(\infty)}(v_i)f^{(\infty)}(v_j)]
         &=\sum_{p=1}^n \Pr_E[\mathcal{M}_{i,j}^{p,p}]\mu_p^{(0)}
         +\sum_{p\neq q} \Pr_E[\mathcal{M}_{i,j}^{p,q}]\mu_p^{(0)}\mu_q^{(0)}\\
         &=\mu_i\mu_j+\sum_{p=1}^n \mu_p^{(0)}(1-\mu_p^{(0)})Q_{ip}Q_{jp}\\
         &\quad +\sum_{\ell=1}^{n} \Pr[\mathcal{I}_{ij}^\ell]\left(\mu_\ell(1-\mu_\ell)-\sum_{p=1}^n \mu_p^{(0)}(1-\mu_p^{(0)})Q_{\ell p}^2 \right).
     \end{align*}
     Thus
     \begin{align*}
         \s_{ij}&=(1-\mu_i)(1-\mu_j)+\mu_i\mu_j+2\sum_{p=1}^n \mu_p^{(0)}(1-\mu_p^{(0)})Q_{ip}Q_{jp}
          +2\sum_{\ell=1}^{n} \Pr[\mathcal{I}_{ij}^\ell]\left(\mu_\ell(1-\mu_\ell)-\sum_{p=1}^n \mu_p^{(0)}(1-\mu_p^{(0)})Q_{\ell p}^2 \right)
     \end{align*}
     where $\e_i$ is calculated by $\sum_{j=1}^n Q_{ij} \mu_j^{(0)}$ as shown in Eq (\ref{eq: non-iid expectation}).
    With opinion similarities calculated, then we can use partitioned sampling method described in Section \ref{section: solving the OPS problem} to do efficient sampling.
\subsection{Signed \model{} model}\label{section: signed}

In this section, we provide a sketchy analysis of the signed \model{} model, which allows negative edges in the graph.
	Given a weighted directed social graph $G=(V,A)$ where $A$ is the weighted adjacency matrix with $A_{ij}\neq 0$ if and only if edge $(v_i,v_j)$ exists, with $A_{ij}\in \mathbb{R}^+\cup\mathbb{R}^-$ as the weight of edge $(v_i,v_j)$.
	At initial time $t=0$, each node $v_i$ generates its innate opinion $f^{(0)}(v_i)$ from an i.i.d.\ distribution
	with mean $\mu^{(0)}$.
    When $v_i$'s Poisson arrives, it sets its expressed opinion
	$f^{(t)}(v_i)$ to be its own innate opinion $f^{(0)}(v_i)$
	with an inward probability $p_i$, or with probability $1-p_i$ 
	node $v_i$ randomly selects a neighbor $v_j$ of $v_i$ with probability
	proportional to absolute value of the weight of the edge $(v_i,v_j)$, i.e., with
	probability $(1-p_i)|A_{ij}|/\sum_{k=1}^{n} |A_{ik}|$, and sets $f^{(t)}(v_i)$ to $v_j$'s expressed opinion $f^{(t)}(v_j)$ when $A_{ij}>0$ or the opposite of $v_j$'s expressed opinion $1-f^{(t)}(v_j)$ when $A_{ij}<0$.
	Similar to the \model{} model, when $p_i > 0$ for all $i\in [n]$, the signed \model{} model has a unique steady state distribution for the final expressed opinions. 
    For simplifying the notation, we use $f_i$ (and $f_i(t)$) to denote $f^{(\infty)}(v_i)$ (and $f^{(t)}(v_i)$).
    We define $f'_i$ (and $f'_i(t)$) to be $2f_i -1$ (and $2f_i(t)-1$).
    Thus $f'_i(t)=1$ if $f_i(t)=1$, and $f'_i(t)=-1$ if $f_i(t)=0$.
    We list all the notations used in this section in Table \ref{table: notations}. 
    In order to distinguish the source of the randomness, we split the randomness $M$ into two parts: we use $O$ to denote the randomness of innate opinions which are generated by the i.i.d.\ Bernoulli distribution, and use $E$ to denote the randomness from the opinion evolution (These notations are first introduced in the proof of Lemma \ref{lemma:correlation}).
	
	\begin{table*}[t]
        \centering
		\begin{tabular}{ | c | m{9.5cm} |}
			\hline
			Notation & Representation \\ \hline
			$A^+$, $A^-$, $\bar A$ & $A^+$ (resp. $A^-$) is the non-negative adjacency matrix representing positive (resp. negative) edges of $G$, with $A = A^+ - A^-$ and $\bar A = A^+ + A^-$. \\ \hline
			$Q^+$, $Q^-$ & The $(i,j)$-entry of $Q^+$ (resp. $Q^-$) is the probability that the random walk starting from $v_i$ is absorbed by $v'_j$ in $V'$ while the number of walking steps is even (resp. odd). \\ \hline
			$p_{\it same}^+(i,j)$, $p_{\it same}^-(i,j)$ & $p_{\it same}^+(i,j)$ (resp. $p_{\it same}^-(i,j)$) is the probability that two random coalescing walks starting from $v_i$ and $v_j$ are absorbed by the same node in $V'$ while the summation of two walks' steps is even (resp. odd). \\ \hline
			$p_{\it diff}^+(i,j)$, $p_{\it diff}^-(i,j)$ & $p_{\it diff}^+(i,j)$ (resp. $p_{\it diff}^-(i,j)$) is the probability that two random coalescing walks starting from $v_i$ and $v_j$ are absorbed by different nodes in $V'$ while the summation of two walks' steps is even (resp. odd). \\ \hline
			$f_i(t)$, $f'_i(t)$, $f_i$, $f'_i$ & $f_i(t)$ and $f'_i(t)$ are both representing $v_i$'s expressed opinion at time $t$. The difference is $f_i(t)\in\{0,1\}$ and $f'_i(t)\in\{-1,1\}$. $f'_i(t)$ can be obtained by $2f_i(t)-1$. $f_i$ and $f'_i$ represent $v_i$'s final expressed opinion ($t\rightarrow\infty$).\\ \hline
			${\cal I}_{ij}^{l+}$, ${\cal I}_{ij}^{l-}$ & ${\cal I}_{ij}^{l+}$ (resp. ${\cal I}_{ij}^{l-}$) is the event that two random walks starting from $v_i$ and $v_j$ eventually meet and the first node they meet at is $v_l\in V$ while the summation of the steps they have taken is even (resp. odd).\\ \hline
			${\cal H}_{ij}^{k+}$, ${\cal H}_{ij}^{k-}$ & ${\cal H}_{ij}^{k+}$ (resp. ${\cal H}_{ij}^{k-}$) is the event that two coalescing random walks starting from $v_i$ and $v_j$ are both absorbed by node  $v_k^\prime$ without meeting each other at a node in $V$ while the summation of two walks' steps is even (resp. odd).\\ \hline
		\end{tabular}
        \caption{Notations of the signed VIO model}\label{table: notations}
	\end{table*}
    
    The expected final expressed opinions can be calculated by (similar to the computation of $Q$ in the \model{} model)
    \begin{align*}
        \E_E[\vec f'] &= \sum_{k=0}^\infty ((I-P)D^{-1}A)^k P \vec f'(0)
        =(I-(I-P)D^{-1}A)^{-1} P \vec f'(0)
    \end{align*}
    where $\vec f'=(f_1',\cdots,f_n')^T$ and $\vec f'(0)=(f_1'(0),\cdots,f_n'(0))^T$.
%    Thus
%    \begin{align*}
%        \E_M[\vec f']&=\E_O\E_E[\vec f']\\
%        &=(I-(I-P)D^{-1}A)P \E_O[\vec f'(0)]\\
%        &=(2\mu^{(0)}-1)(I-(I-P)D^{-1}A) P\cdot \vec 1.
%    \end{align*}
    Notice that the $\E_E[\vec f']$ can also be written as  $(Q^+-Q^-)\vec f'(0)$ for any $\vec f'(0)$, so we have 
    $$Q^+-Q^-=(I-(I-P)D^{-1}A)^{-1}P.$$
    In the analysis of the unsigned \model{} model, we have
    $$Q^++Q^-=(I-(I-P)D^{-1}\bar A)^{-1}P.$$
    Therefore, we can obtain
    $$Q^+=(I-(I-P)D^{-1} A^+)^{-1}P,$$
    and $$Q^-=(I-(I-P)D^{-1} A^-)^{-1}P.$$

    Notice that for any $i\in[n]$, 
    \begin{equation*}
        \mu_i=\E_M[f_i]=(\E_M[f_i']+1)/2.
    \end{equation*}
    %Then, we focus on computing the opinion similarity between any pair of nodes $(v_i,v_j)$, 
    %which is the probability that two random walkers starting from $v_i$ and $v_j$ have the same final expressed opinions.
    Recall Eq (\ref{eq: E[fifj]}) that
    \begin{align*}
        \s_{ij}&=2\E_M[f_i f_j]+1 -\e_i-\e_j \\
        &=2\E_M[\frac{1+f_i'}{2}\frac{1+f_j'}{2}]+1-\frac{1+\E_M[f_i']}{2}-\frac{1+\E_M[f_j']}{2}\\
        &=\frac{1}{2}\E_M[f_i' f_j']+\frac{1}{2}.\numberthis\label{eq: s1}
    \end{align*}
    Next, we focus on calculating $\E_M[f'_i f'_j]$. 
    We consider the following two cases: a) two coalescing random walkers from $v_i$ and $v_j$ walking to the same node in $V'$, and b) two coalescing random walkers from $v_i$ and $v_j$ walking to different nodes in $V'$. 
    Notice that if two random walkers end at the same node $v_k'\in V'$, thus $\E_O[f'_k f'_k]=1$ for any $k\in [n]$; and if two random walkers end at the different nodes  $v_k'$ and $v_l'$, thus $\E_O[f'_k f'_l]=(2\mu^{(0)}-1)^2$ for any $k,l\in [n]$.
    Therefore,
    \begin{align*}
        \E_M[f'_i f'_j]
        =\left[p_{\it same}^+(i,j)-p_{\it same}^-(i,j)\right]\cdot 1
        +\left[p_{\it diff}^+(i,j)-p_{\it diff}^-(i,j)\right]\cdot (2\mu^{(0)}-1)^2.\numberthis\label{eq: s2}
    \end{align*}    
    
    Next, we calculate $p_{\it same}^+(i,j)-p_{\it same}^-(i,j)$ and $p_{\it diff}^+(i,j)-p_{\it diff}^-(i,j)$ separately,
    similar to the proof of Lemma \ref{lemma:correlation}.
    
    \quad\\
    (a) $p_{\it same}^+(i,j)-p_{\it same}^-(i,j)$.
    
    If two walkers starting from $v_i$ and $v_j$ walk to the same node in $V'$, there are two cases: one is they meet at some node $v_l\in V$ and then walk together until being absorbed; the other is they do not meet before they end their walks.
    Thus we have
    \begin{align*}
        p_{\it same}^+(i,j)=\sum_{k=1}^n\sum_{l=1}^n\Pr[\mathcal{I}_{ij}^{l+}](Q_{lk}^++Q_{lk}^-)+\sum_{k=1}^n\Pr[\mathcal{H}_{ij}^{k+}],
    \end{align*}
    and
    \begin{align*}
    p_{\it same}^-(i,j)=\sum_{k=1}^n\sum_{l=1}^n\Pr[\mathcal{I}_{ij}^{l-}](Q_{lk}^++Q_{lk}^-)+\sum_{k=1}^n\Pr[\mathcal{H}_{ij}^{k-}].
    \end{align*}
    Consider two non-coalescing random walks starting from $v_i$ and $v_j$ are absorbed by the same node  $v_k'\in V'$ while the summation of two walks' steps is even:
    \begin{align*}
        Q_{ik}^+ Q_{jk}^+ + Q_{ik}^- Q_{jk}^- 
        =& \sum_{l=1}^n \Pr[\mathcal I_{ij}^{l+}]\left({Q_{lk}^+}^2+{Q_{lk}^-}^2\right)
        +\sum_{l=1}^n \Pr[\mathcal I_{ij}^{l-}] \cdot 2 Q_{lk}^+Q_{lk}^-
        +\Pr[\mathcal{H}_{ij}^{k+}].
    \end{align*}
    Consider two non-coalescing random walks starting from $v_i$ and $v_j$ are absorbed by the same node  $v_k'\in V'$ while the summation of two walks' steps is odd:
    \begin{align*}
        Q_{ik}^+ Q_{jk}^- + Q_{ik}^- Q_{jk}^+ 
        =& \sum_{l=1}^n \Pr[\mathcal I_{ij}^{l-}]\left({Q_{lk}^+}^2+{Q_{lk}^-}^2\right)
        +\sum_{l=1}^n \Pr[\mathcal I_{ij}^{l+}] \cdot 2 Q_{lk}^+Q_{lk}^-
        +\Pr[\mathcal{H}_{ij}^{k-}].
    \end{align*}
    Thus
    \begin{align*}
        p_{\it same}^+(i,j)-p_{\it same}^-(i,j)
        = \sum_{k=1}^n\sum_{l=1}^n \left( \Pr[\mathcal I_{ij}^{l+}] - \Pr[\mathcal I_{ij}^{l-}] \right) \left[ 1- (Q_{lk}^+ - Q_{lk}^-)^2 \right]
          + \sum_{k=1}^n (Q_{ik}^+ -Q_{ik}^-)(Q_{jk}^+ - Q_{jk}^-).\numberthis\label{eq: s3}
    \end{align*}
    
    \quad\\
    (b) $p_{\it diff}^+(i,j)-p_{\it diff}^-(i,j)$.
    
    Consider two non-coalescing random walks starting from $v_i$ and $v_j$ are absorbed by different nodes in $V'$ while the summation of two walks' steps is even:
    \begin{align*}
        \sum_{a \neq b}Q_{ia}^+Q_{jb}^+ + Q_{ia}^-Q_{jb}^-
        =\sum_{a \neq b}\sum_{l=1}^n \Pr[\mathcal I_{ij}^{l+}](Q_{la}^+Q_{lb}^+ + Q_{la}^-Q_{lb}^-)
        +\sum_{a \neq b}\sum_{l=1}^n \Pr[\mathcal I_{ij}^{l-}](Q_{la}^+Q_{lb}^- + Q_{la}^-Q_{lb}^+)
        +p_{\it diff}^+(i,j).
    \end{align*}
    Consider two non-coalescing random walks starting from $v_i$ and $v_j$ are absorbed by different nodes in $V'$ while the summation of two walks' steps is odd:
    \begin{align*}
        \sum_{a \neq b}Q_{ia}^+Q_{jb}^- + Q_{ia}^-Q_{jb}^+
        =\sum_{a \neq b}\sum_{l=1}^n \Pr[\mathcal I_{ij}^{l+}](Q_{la}^+Q_{lb}^- + Q_{la}^-Q_{lb}^+)
        +\sum_{a \neq b}\sum_{l=1}^n \Pr[\mathcal I_{ij}^{l-}](Q_{la}^+Q_{lb}^+ + Q_{la}^-Q_{lb}^-)
        +p_{\it diff}^-(i,j).
    \end{align*}
    Thus
    \begin{align*}
        p_{\it diff}^+(i,j)-p_{\it diff}^-(i,j)
        =\sum_{a \neq b} (Q_{ia}^+-Q_{ia}^-)(Q_{jb}^+-Q_{jb}^-)
        -\sum_{a \neq b} \left( \Pr[\mathcal I_{ij}^{l+}] 
        - \Pr[\mathcal I_{ij}^{l-}] \right) (Q_{la}^+-Q_{la}^-)(Q_{lb}^+-Q_{lb}^-).\numberthis\label{eq: s4}
    \end{align*}
    
    In order to obtain $p_{\it same}^+(i,j)-p_{\it same}^-(i,j)$ and $p_{\it diff}^+(i,j)-p_{\it diff}^-(i,j)$ by Eq (\ref{eq: s3}, \ref{eq: s4}), we need to compute $\Pr[\mathcal I_{ij}^{l+}]$ and $\Pr[\mathcal I_{ij}^{l-}]$.
    Let $\Pr[\mathcal I_{ij}^{l}]=\Pr[\mathcal I_{ij}^{l+}] + \Pr[\mathcal I_{ij}^{l-}]$, $\Delta[\mathcal I_{ij}^{l}]=\Pr[\mathcal I_{ij}^{l+}] - \Pr[\mathcal I_{ij}^{l-}]$.
    
    Similar to the proof of Lemma \ref{lemma:walkparam}, we have
    \begin{align*}
        \Pr[\mathcal I_{ij}^{l+}] =& \sum_{a=1}^n \frac{\lambda_i (1-p_i) |A_{ia}|}{2 (\lambda_i+\lambda_j)d_i } \left( \Pr[\mathcal I_{aj}^l] + {E_{ia}} \Delta[\mathcal I_{aj}^{l}] \right) 
         +\sum_{b=1}^n \frac{\lambda_j (1-p_j) |A_{jb}|}{2 (\lambda_i+\lambda_j)d_j } \left( \Pr[\mathcal I_{ib}^l] + {E_{jb}} \Delta[\mathcal I_{ib}^{l}] \right),
    \end{align*}
    and
    \begin{align*}
        \Pr[\mathcal I_{ij}^{l-}] =& \sum_{a=1}^n \frac{\lambda_i (1-p_i) |A_{ia}|}{2 (\lambda_i+\lambda_j)d_i } \left( \Pr[\mathcal I_{aj}^l] - {E_{ia}} \Delta[\mathcal I_{aj}^{l}] \right) 
         +\sum_{b=1}^n \frac{\lambda_j (1-p_j) |A_{jb}|}{2 (\lambda_i+\lambda_j)d_j } \left( \Pr[\mathcal I_{ib}^l] - {E_{jb}} \Delta[\mathcal I_{ib}^{l}] \right).
    \end{align*}
    where 
    \begin{align*}
        E_{ij}=\begin{cases}
            1, & A_{ij}>0; \\
            -1, & A_{ij}<0; \\
            0, & A_{ij}=0.
        \end{cases}
    \end{align*}
    Thus, we can get the recursive equations of $\Pr[\mathcal I_{ij}^{l}]$ and $\Delta[\mathcal I_{ij}^{l}]$:
    \begin{align*}
    \Pr[\mathcal I_{ij}^{l}] =& \sum_{a=1}^n \frac{\lambda_i (1-p_i) |A_{ia}|}{ (\lambda_i+\lambda_j)d_i }  \Pr[\mathcal I_{aj}^l]  
     +\sum_{b=1}^n \frac{\lambda_j (1-p_j) |A_{jb}|}{ (\lambda_i+\lambda_j)d_j }  \Pr[\mathcal I_{ib}^l],
    \end{align*}
    and
    \begin{align*}
    \Delta[\mathcal I_{ij}^{l}] =& \sum_{a=1}^n \frac{\lambda_i (1-p_i) |A_{ia}|}{ (\lambda_i+\lambda_j)d_i } E_{ia} \Delta[\mathcal I_{aj}^l]  
     +\sum_{b=1}^n \frac{\lambda_j (1-p_j) |A_{jb}|}{ (\lambda_i+\lambda_j)d_j } E_{jb} \Delta[\mathcal I_{ib}^l].
    \end{align*}
    
    After solving the above two recursive equations, we get $\Pr[\mathcal I_{ij}^{l}]$ and $\Delta[\mathcal I_{ij}^{l}]$,
    and then we can get $\Pr[\mathcal I_{ij}^{l+}]= \frac{\Pr[\mathcal I_{ij}^{l}] + \Delta[\mathcal I_{ij}^{l}]}{2}$ and $\Pr[\mathcal I_{ij}^{l-}]= \frac{\Pr[\mathcal I_{ij}^{l}] - \Delta[\mathcal I_{ij}^{l}]}{2}$.
    
    Above all, we can get opinion similarities between any pair of nodes by Equation (\ref{eq: s1}, \ref{eq: s2}, \ref{eq: s3}, \ref{eq: s4}).

    \endgroup
}

\end{document}